\RequirePackage{fix-cm}
\documentclass[twocolumn]{svjour3}          % twocolumn
\smartqed  % flush right qed marks, e.g. at end of proof
\usepackage{graphicx}
\usepackage[cmex10]{amsmath}
\usepackage{setspace}
\usepackage{comment}
\usepackage{placeins}
\usepackage{afterpage}
\usepackage[utf8]{inputenc}
\usepackage{graphicx}
\usepackage{amssymb}
\usepackage{cite}
\usepackage{color}
\usepackage{caption}
\usepackage{amsmath}
\usepackage{tabu}
\usepackage{array}
\usepackage{booktabs}
\usepackage{amssymb}
\usepackage{threeparttable}
\usepackage[linesnumbered,ruled]{algorithm2e}
\usepackage{float}
\usepackage{stfloats}
\usepackage{comment}
\usepackage{lipsum}
\usepackage[english]{babel}
\usepackage[T1]{fontenc}
\usepackage{mwe}    % loads »blindtext« and »graphicx«
\usepackage{subfig}

\let\oldnl\nl% Store \nl in \oldnl
\newcommand{\nonl}{\renewcommand{\nl}{\let\nl\oldnl}}% Remove line number for one line
\newtheorem{defn}{Definition}
\newtheorem{exmp}{Example} % same for example numbers
\newtheorem{thm}{Theorem}

\newtheorem{prop}{Property}

\begin{document}

\title{Supporting Secure Dynamic Alert Zones Using Searchable Encryption and Graph Embedding%\thanks{Grants or other notes
}


\author{Sina Shaham         \and
        Gabriel Ghinita     \and 
        Cyrus Shahabi%etc.
}


\institute{S. Shaham \at
              Department of Computer Science\\ University of Southern California\\
              Los Angeles, CA, USA\\
              \email{sshaham@usc.edu}           %  \\
        \and
            G. Ghinita \at
              College of Science and Engineering\\ Hamad Bin Khalifa University\\
               Qatar Foundation, Doha, Qatar\\
              \email{gghinita@hbku.edu.qa}      
           \and
            C. Shahabi \at
              Department of Computer Science\\ University of Southern California\\
              Los Angeles, CA, USA\\
              \email{shahabi@usc.edu}    
}

\date{Received: date / Accepted: date}

\maketitle

\begin{abstract}
Location-based alerts have gained increasing popularity in recent years, whether in the context of healthcare (e.g., COVID-19 contact tracing), marketing (e.g., location-based advertising), or public safety.
However, serious privacy concerns arise when location data are used in clear in the process. Several solutions employ {\em Searchable Encryption (SE)} to achieve {\em secure} alerts directly on encrypted locations. While doing so preserves privacy, the performance overhead incurred is high. We focus on a prominent SE technique in the public-key setting -- {\em Hidden Vector Encryption (HVE)}, and propose a graph embedding technique to encode location data in a way that significantly boosts the performance of processing on ciphertexts. We show that finding the optimal encoding is NP-hard, and provide several heuristics that are fast and obtain significant performance gains. Furthermore, we investigate the more challenging case of dynamic alert zones, where the area of interest changes over time. Our extensive experimental evaluation shows that our solutions can significantly improve computational overhead compared to existing baselines.
\end{abstract}

\section{Introduction}\label{Introduction}

Location data play an important part in offering customized services to mobile users. Whether they are used to find nearby points of interest, to offer location-based recommendations, or to locate friends situated in proximity to each other, location data significantly enrich the type of interactions between users and their favorite services. However, current service providers collect location data in clear, and often share it with third parties, compromising users' privacy. Movement data can disclose sensitive details about an individual's health status, political orientation, alternative lifestyles, etc. Hence, it is important to support such location-based interactions while protecting privacy.

Our focus is on {\em secure alert zones}, a type of location-based service where users report their locations in encrypted form to a service provider, and then they receive alerts when an event of interest occurs in their proximity. This operation is very relevant to contact tracing, which is proving to be essential in controlling pandemics, e.g., COVID-19. It is important to determine if a mobile user came in close proximity to an infected person, or to a surface that has been exposed to the virus, but at the same time one must prevent against intrusive surveillance of the population. More applications of alert zones include public safety notifications (e.g., active shooter), and commercial applications (e.g., notifying mobile users of nearby sales events).

{\em Searchable Encryption (SE)} \cite{song2000,boneh2007conjunctive,HXT18} is very suitable for implementing secure alert zones. Users encrypt their location before sending it to the service provider using a special kind of encryption, which allows the evaluation of predicates directly on ciphertexts. However, the underlying encryption functions are not specifically designed for geospatial queries, but for arbitrary keyword or range queries. As a result, a data mapping step is typically performed to transform spatial queries to the primitive operations supported on ciphertexts. Due to this translation, the performance overhead can be significant. Some solutions use {\em Symmetric Searchable Encryption (SSE)} \cite{song2000,curtmola2011searchable,HXT18}, where a trusted entity knows the secret key of the transformation, and collects the location of all users before encrypting them and sending the ciphertext to the service provider. While the performance of SSE can be quite good, the system model that requires mobile users to share their clear text locations with a trusted service is not adequate from a privacy perspective, since it still incurs a significant amount of disclosure.

To address the shortcomings of SSE models, the work in \cite{boneh2007conjunctive} introduced the novel concept of {\em Hidden Vector Encryption (HVE)}, which is an {\em asymmetric} type of encryption that allows direct evaluation of predicates on top of ciphertext. Each user encrypts her own location using the {\em public} key of the transformation, and no trusted component that accesses locations in clear is required. This approach has been considered in the location context in~\cite{ghinita2014efficient},~\cite{nguyen2019privacy}, with encouraging results. However, the performance overhead of HVE in the spatial domain remains high. Motivated by this fact, we study techniques to reduce the computational overhead of HVE. Specifically, we derive special types of spatial data mapping using graph embeddings, which allow us to express spatial queries with predicates that are less computationally-intensive to evaluate.

In existing HVE work for geospatial data~\cite{ghinita2014efficient},~\cite{nguyen2019privacy}, the data domain is partitioned into a hierarchical data structure, and each node in this structure is assigned a binary string identifier. The binary representation of each node plays an important part in the query encoding, and it influences the amount of computation that needs to be executed when evaluating predicates on ciphertexts. However, the impact of the specific encoding is not evaluated in-depth. Our approach embeds the geospatial data domain to a high-dimensional hypercube, and then it applies graph embedding \cite{chandrasekharam1994genetic} techniques that directly target the reduction of computation overhead in the predicate evaluation step. Finally, no existing work considers the case of alert zones that change over time. Support for dynamic alert zones is very important, given that in most use case scenarios, phenomena of interest evolve over time (e.g., places visited by COVID carriers, area affected by a gas leak, etc). Our work tackles this important challenge\footnote{This submission is an extended version of the work in \cite{dbsec}. Additional contributions include: a complexity analysis of Gray-based approaches (Section~\ref{Sec: Complexity Analysis}); novel techniques for supporting dynamic alert zones~(Section~\ref{Advance Modeling of Alert Zones}); and an updated experimental section, including empirical evaluation of dynamic alert zones algorithms.}.

Our specific contributions are:

\begin{itemize}
 
\item We introduce a novel transformation of the spatial data domain based on graph embedding that is able to model accurately the performance overhead incurred when running HVE queries for spatial predicates;

\item We transform the problem of minimizing HVE computation to a graph problem, and show that the optimal solution is NP-hard;

\item We devise several heuristics that can solve the problem efficiently in the embedded space, while reducing significantly the computational overhead;

\item We propose models that take into account the spatial and temporal evolution of alert zones, and choose encodings that improve performance under dynamic conditions;

\item We perform an extensive experimental evaluation which shows that the proposed approaches are able to halve the performance overhead incurred by HVE when processing spatial queries.
    
\end{itemize}

The rest of the paper is organized as follows: Section 2 introduces necessary background on the system model (an HVE primer is given in Appendix~\ref{sec:app}). Section 3 provides the details of the proposed graph embedding transformation. Section 4 introduces several heuristic algorithms that solve the problem efficiently. Section 5 focuses on modeling of dynamic alert zones, and on advanced encodings under changing conditions. Section 6 evaluates thoroughly the proposed approach on real-life datasets. We survey related work in Section 7 and conclude in Section 8.

%%%%%%%%%%%%%%%%%%%%%%%%%%%%%%%%%%%%%%%%%
\section{Background}\label{Sec: Background}

\subsection{System Model}\label{Sec: System Model}

Consider a [0,1]$\times$[0,1] spatial data domain divided into $n$ non-overlapping partitions, denoted as 
\begin{equation}
    \mathcal{V}=\{ v_1,\,,v_2,..., v_{n} \}.
\end{equation}
We use the term {\em cell} to refer to partitions, which can have an arbitrary size and shape. An example of such a partitioning is provided in Fig.~\ref{Fig: Sample grid}. 
The system architecture of location-based alert system is represented in Fig.~\ref{fig:system model}, and consists of three types of entities:

\newcommand{\rvec}{\mathrm {\mathbf {r}}} 
\begingroup
\begin{table}
\caption {Summary of notations.} 
\centering
\begin{tabular}{>{\arraybackslash}m{2cm} >{\arraybackslash}m{5.8cm} }
\hline\hline
 Symbol  & Description\\    \hline
  $n$ & Number of cells \\
  $k$ & Length of HVE index \\
  $\mathcal{V} = \{ \bigcup v_i\}$ & Set of all cells\\
  $C_j$ & Encrypted location of user $j$\\
  $TK$ & Token $j$\\
  $M_j$ & message of user $j$\\
% \hline $p(.)$ & This function returns the probability of cell or set of cells being alerted. \\
  $G_1(\mathcal{C}, \mathcal{E}_1)$ & $k$-cube; vertices $\mathcal{C} = \{ \bigcup c_i\}$; edges $\mathcal{E}_1$\\
 $G_2(\mathcal{V}, \mathcal{E}_2)$ & Complete graph; vertices $\mathcal{V} = \{ \bigcup v_i\}$;  edges $\mathcal{E}_2$ \\
 $(\mathcal{H}_1,\mathcal{H}_2,\mathcal{E}_3)$ & Bipartite graph; node sets $\mathcal{H}_1$ and $\mathcal{H}_2$; edges $\mathcal{E}_3$ \\
$\mathcal{D}_{i|c_j}$ & Set of nodes in $\mathcal{C}$ with Hamming distance $i$ from $c_j$ \\
$\mathcal{L}_x$ & Set of all complete $x$-bit gray cycles, $\mathcal{L}_x = \{ \bigcup l_i \}$ \\
$\mathcal{S}_n$& state space with $n$ cells\\
$Q_n$& state transition matrix with $n$ cells\\
Pois($\lambda$)& Poisson distribution; occurrence rate $\lambda$\\
$\boldsymbol{s}$& Stationary distribution vector\\
$\textbf{\underline{i}}$& state $i$\\
\hline\hline
\end{tabular}
\label{tab:table1}
\end{table}
\endgroup

\begin{enumerate}
\item 
{\bf Mobile Users} subscribe to the alert system and periodically submit encrypted location updates.
\item
The {\bf Trusted Authority (TA)} is a trusted entity that decides which are the alert zones, and creates for each zone a search {\em token} that allows to check privately if a user location falls within the alert zone or not.
\item
The {\bf Server (S)} is the provider of the alert service. It receives encrypted updates from users and search tokens from TA, and performs the predicate evaluation to decide whether encrypted location $C_i$ of user $i$ falls within alert zone $j$ represented by token $TK_j$. If the predicate holds, the server learns message $M_i$ encrypted by the user, otherwise it learns nothing.
\end{enumerate}

Table~\ref{tab:table1} summarizes the notations used throughout the manuscript.

\begin{figure*}[t]
\centering
\includegraphics[scale=.55]{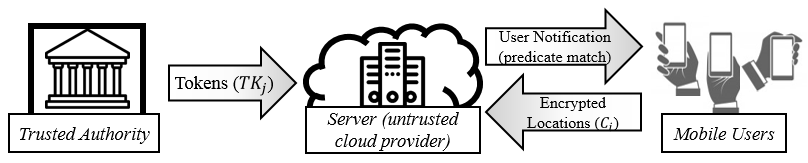}
\hspace{1em}
\centering
\caption{Location-based alert system.}
\label{fig:system model}
\end{figure*}

\begin{figure}[t]
\centering
\includegraphics[scale=.45]{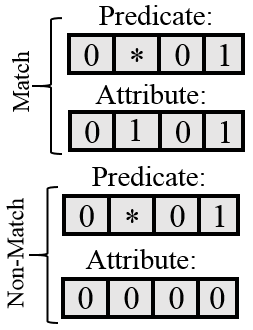}
\hspace{1em}
\centering
\caption{HVE evaluation.}
\label{fig:math-nonmatch}
\end{figure}

The system supports location-based {\em alerts}, with the following semantics: a {\em Trusted Authority (TA)} designates a subset of cells as an {\em alert zone}, and all the users enclosed by those cells must be notified. The TA can be, for instance, the Center for Disease Control (CDC), who is monitoring cases of a pandemic, and wishes to notify users who may have been affected; or, the TA can be some commercial entity that the users subscribe to, and who notifies users when a sales event occurs at selected locations.

The {\em privacy requirement} of the system dictates that the server must not learn any information about the user locations, other than what can be derived from the match outcome, i.e., whether the user is in a particular alert zone or not. In case of a successful match, the server $S$ learns that user $u$ is enclosed by zone $z$. In case of a non-match, the server $S$ learns only that the user is outside the zone $z$, but no additional location information. Note that, this model is applicable to many real-life scenarios. For instance, users wish to keep their location private most of the time, but they want to be immediately notified if they enter a zone where their personal safety may be threatened. Furthermore, the extent of alert zones is typically small compared to the entire data domain, so the fact that $S$ learns that $u$ is {\em not} within the set of alert zones does not disclose significant information about $u$'s location.
The TA can be an organization such as CDC, or a city's public emergency department, which is trusted not to compromise user privacy, but at the same time does not have the infrastructure to monitor a large user population, and outsources the service to a cloud provider.

\begin{figure*}[t]
\centering
	\subfloat[Sample grid.\label{Fig: Sample grid}]{%
	\includegraphics[scale = 0.5]{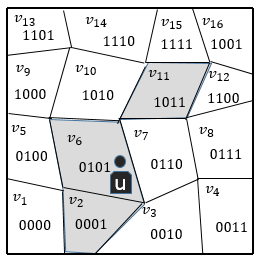}
	}	
	\subfloat[Graph $G_1(\mathcal{C}, \mathcal{E}_1)$.\label{Fig: Graph G1}]{%
		\includegraphics[scale = 0.46]{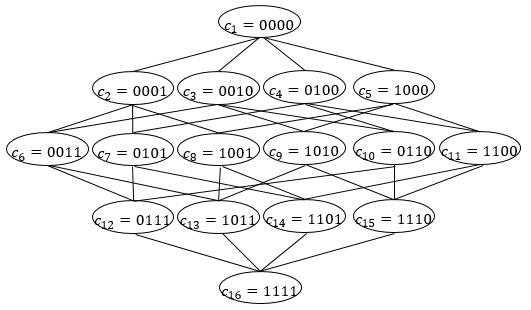}	
	}
	\subfloat[Graph $G_2(\mathcal{V}, \mathcal{E}_2)$.\label{Fig: Graph G2}]{%
		\includegraphics[scale=.5]{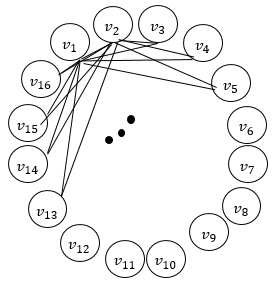}
	}
	\caption{An example of embedding graphs generated based on a sample grid.}
\end{figure*}

\subsection{Problem Statement}\label{Sec: Problem Statement}

Prior work~\cite{ghinita2014efficient,nguyen2019privacy} assumed that all cells are equally likely be in an alert zone. However, that is not the case in practice. Some parts of the data domain (e.g., denser areas of a city) are more likely to become alert zones. 
The cost of encrypted alert zone enclosure evaluation is given by the number of operations required to apply HVE matching at the service provider.
As we discuss in our HVE primer in Appendix A, the evaluation cost is directly proportional to the number of non-star bits in the tokens. 
Armed with knowledge about the likelihood of cells to be part of an alert zone, one can create superior encodings that reduce processing overhead.

Our goal is to find an enhanced encoding that reduces non-star bits for a given set of alert zone tokens. 
Denote by $p(v_i)$ the probability of cell $v_i$ being part of an alert zone. The {\em mutual} probability of multiple cells indicates how likely they are to be part of the {\em same} alert zone. Given individual cell probabilities, the mutual probability of a set of $i$ cells $\mathcal{L} = \{ v_1',\,,v_2',..., v_{i}' \}$ is calculated as:

\begin{equation}\label{base}
	p(\mathcal{L}) = \prod_{j=1}^{i} p(v'_j).
\end{equation}

The problem we study is formally presented as follows:
\begin{problem}\label{Problem: 1}
Find an encoding of the grid that on average reduces the number of non-star bits in the tokens generated from alert zone cells.
\end{problem}

In the above formulation, the correlation between cells becoming part of an alert zone is assumed to be negligible. In essence, the assumption is that cells are independent in time and space (in Section~\ref{Advance Modeling of Alert Zones},  we provide an advanced modeling of the correlation of alert zones over space and time). %Note that, our proposed algorithms can work regardless of assumptions as the probabilities devised from modeling are input to the algorithms. Hence, it is application-specific to decide which space and time modeling provide an enhanced approximation without imposing high computation complexity.

%%%%%%%%%%%%%%%%%%%%%%%%%%%

\section{Location Domain Mapping through Graph Embedding} \label{sec:ProblemEmbedding}

Our approach minimizes the number of non-star bits in alert zone tokens by modeling the data domain partitioning as an embedding problem of a $k$-cube onto a complete graph. We denote a $k$-cube as $G_1(\mathcal{C}, \mathcal{E}_1)$, where $\mathcal{C}=\{ c_1,\,,c_2,..., c_{n} \}$ and $c_i=\{0,1\}^k$. Fig.~\ref{Fig: Graph G1} illustrates a $k$-cube generated based on the sample partitioning in Fig~\ref{Fig: Sample grid}. In $G_1$, two nodes $c_i$ and $c_j$ are connected if their \textit{Hamming distance} is equal to one. We refer to such a bit as \textit{Hamming bit}.
\begin{defn}{(Hamming Distance and Bits).}\label{Hamming distance and bits}
	The Hamming distance between two indices $c_i$ and $c_j$ in $G_1(\mathcal{C}, \mathcal{E}_1)$ is the minimum number of substitutions required to transform $c_i$ to $c_j$, denoted by the function $d_h(.)$. We refer to the bits that need to be substituted as the {\em Hamming bits} of the indices. 
\end{defn}

\begin{exmp}
The Hamming distance between indices $c_1=0100$ and $c_2 = 0010$ is two ($d_h(c_i,c_j)=2$), and the Hamming bits are the second and third most significant bits of the indices. 
\end{exmp}

The second graph required to formulate the problem of minimizing the number of non-stars is a complete graph generated by all cells in the partitioning, denoted by $G_2(\mathcal{V}, \mathcal{E}_2)$. The set $\mathcal{V}$ represents the nodes corresponding to cells, and an undirected edge connects every two nodes in $G_2$.

Note that, every token (including those containing stars), can be related to several cycles on the $k$-cube. For example, token 00** represents four indices 0000, 0001, 0010, 0011, which correspond to cycles $(c_1,c_2,c_6,c_3)$ and $(c_1,c_3,c_6,c_2)$ on the $k$-cube in Fig.~\ref{Fig: Graph G1}. Unfortunately, there is no one-to-one correspondence between the tokens and the cycles. In particular, for a larger number of stars, there exist several cycles representing the same token. To generate a one-to-one correspondence, we incorporate {\em Binary-Reflected Gray (BRG)} encoding on the $k$-cube to create unique cycles corresponding to tokens.

\begin{defn}(BRG path on $k$-cube).\label{Defn: gray path}
	A BRG path between two nodes with non-zero Hamming distance is defined as the path on the $k$-cube going from one node to another based on BRG coding on Hamming bits.
\end{defn}

As an example, the Hamming bits between 0001 and 1000 are the least and most significant bits, and the BRG path connecting them on the $k$-cube in Fig.\ref{Fig: Graph G1} includes indices 0001, 1001, and 1000 in the given order. One can see that as the BRG codes are unique, the BRG path between two indices on the $k$-cube is also unique. This characteristic of BRG paths is formulated in Lemma~\ref{Thm: gray path lemma}.

\begin{lemma}\label{Thm: gray path lemma}
A BRG path between two nodes on a $k$-cube is unique.
\end{lemma}

\begin{proof}
The uniqueness of the path between two nodes on the $k$-cube follows from the uniqueness of BRG code, as only one such path can be constructed. 
\end{proof}

\begin{defn}(Complete $x$-bit BRG cycle).\label{Defn: gray cycle}
	Given a $k$-cube, a complete $x$-bit BRG cycle is a cyclic BRG path with the length of $2^x$, in which only $x$ bits are affected. We denote the set of all possible complete $x$-bit BRG cycles by $\mathcal{L}_x = \{ \bigcup l_i \}$. 
\end{defn}

\begin{exmp}\label{Example: 2}
In Fig.~\ref{Fig: Graph G1}, token *0** entails eight indices 0000, 0001, 0011, 0010, 1010, 1011, 1001, 1000. This token maps uniquely to the complete 3-bit BRG cycle on the 4-cube with nodes $(c_1,\,c_2,\,c_6,\,c_3,\,c_{9},\,c_{13},\,c_8,\,c_5)$ and start point $c_1$. 
\end{exmp}

We can uniquely associate a token to a cycle on the $k$-cube.  Consider a token with $k$ bits and $x$ stars. This token is mapped to a complete $x$-bit BRG cycle on the $k$-cube, starting from a node in which all the star bits are set to zero. Such a cycle is unique and has a length of $2^x$. Based on this mapping, every token is associated with a unique cycle on the $k$-cube, and every complete $x$-bit BRG cycle is mapped to a unique token with $x$-stars. Therefore, there is a one-to-one correspondence between tokens and complete BRG cycles. 
The formulation of Problem~\ref{Problem: 1} based on graph embedding can be written as follows:

\begin{problem}\label{Problem: 2}
    Given two graphs $G_1(\mathcal{C}, \mathcal{E}_1)$ and $G_2(\mathcal{V}, \mathcal{E}_2)$, find a mapping function $\mathcal{F}: G_1\rightarrow G_2$ with the objective to
    \begin{align}\label{maximizing1}
    Maximize\{\sum_{i=1}^{k} p(\mathcal{L}_i)\}. 
    \end{align}
\end{problem}

\subsection{Gray Optimizer (GO)}\label{Sec: Gray Optimizer}

The problem of embedding a complete graph within a minimized size $k$-cube has been shown to be NP-hard~\cite{chandrasekharam1994genetic}. We develop an heuristic algorithm called \textit{Gray Optimizer} that solves Problem~\ref{Problem: 2}. Consider an initial node of the complete graph $v_r \in \mathcal{V}$, and without loss of generality assume that it is assigned to index $c_1$. We refer to nodes in $\mathcal{G}_1$ interchangeably using their vertex id or binary index. The optimization problem can be formulated as follows.

\begin{problem}\label{Problem: 3}
    Given two graphs $G_1(\mathcal{C}, \mathcal{E}_1)$ and $G_2(\mathcal{V}, \mathcal{E}_2)$, and the node $v_r\in \mathcal{V}$ assigned to index $c_1$, find a mapping function $\mathcal{F}: G_1\rightarrow G_2$ that  
    \begin{align}\label{maximizing2}
    Maximize\{ \sum_{i=1}^{k} p(\mathcal{L}_i|v_r)\}. 
    \end{align}
    \end{problem}

Problem~\ref{Problem: 2} requires an assignment of vertices in $G_2$ to the nodes of $G_1$ such that the probability of complete BRG cycles is maximized; whereas Problem~\ref{Problem: 3} seeks to maximize the probability of cycles with respect to a particular node, in this case $v_r$, which is assigned to the index $c_1$. A reasonable candidate for assignment to $c_1$ is the cell with the highest probability, as it is most likely to be part of an alert zone. To solve this problem, we propose the heuristic in Algorithm~\ref{Algo: Mapping Algorithm}.
The input of the algorithm is the root index $c_1\in G_1$, the root node $v_r\in G_2$ (also called seed) and the graphs $G_1$ and $G_2$. 

\begin{algorithm}[tbh]
	\DontPrintSemicolon 
	\SetKwInOut{Input}{Input}\SetKwInOut{Output}{output}
	\Input{$G_1$; $G_2$; $c_1$; $v_r$}
	Sort nodes in $G_2$ based on probabilities\\ 
    Assign $v_r$ to $c_1$\\
	\For {$i\, \textrm{in}\, [1:k]$} {
	Initialize $\mathcal{H}_1, \mathcal{H}_2=\emptyset$\\  
    $\mathcal{H}_1 \leftarrow \{ \binom{k}{i}$ non-assigned nodes in $G_2$ with the highest probability$\}$\\
   
    \For {$c_j \in \mathcal{D}_{i|c_1}$}{
    Calculate $ p(l_j/c_j) = \prod_{v\in l_j/c_j} p(v)$\\
    $\mathcal{H}_2 \leftarrow p(l_j/c_j)$\\
    }
    Sort nodes in $\mathcal{H}_2$ \\
   Match vertices in $\mathcal{H}_1$ to $\mathcal{H}_2$\\
	}
	\caption{Gray Optimizer.}
	\label{Algo: Mapping Algorithm}
\end{algorithm}

Denote by $\mathcal{D}_{i|c_1}$ the set of nodes on $\mathcal{C}$ that have a Hamming distance of $i$ from $c_1$. Note that $\mathcal{D}_{i|c_1}$ includes $\binom{k}{i}$ nodes, each one having a Hamming distance of $i$ from $c_1$. The overall assignment structure is as follows: first, Algorithm~\ref{Algo: Mapping Algorithm} assigns the remaining nodes of $\mathcal{V}$ of the graph $\mathcal{G}_2$ to nodes in $\mathcal{D}_{1|c_1}$. After assignment of all nodes in $\mathcal{D}_{1|c_1}$, the algorithm assigns the nodes in $\mathcal{D}_{2|c_1}$ and follows the same process until all nodes are assigned ($\mathcal{D}_{1|c_1}$ to $\mathcal{D}_{k|c_1}$). An initial sorting of nodes in $\mathcal{V}$ is conducted at the start of the algorithm, and is used throughout the assignment process to reduce the computation complexity.

The assignment objective in stage $i$ of the process is to maximize $p(\mathcal{L}_i|v_r)$. 

Note that~(\ref{maximizing2}) can be written as: 
\begin{align}\label{maximizing3}
\sum_{i=1}^{k} Maximize\{p(\mathcal{L}_i|v_r)\}. 
\end{align}
where $p(\mathcal{L}_i|v_r)$ represents the probability of all complete $i$-bit BRG cycles that include $c_1$ ($v_r\rightarrow c_1$). Denote such a cycle by $l$. Based on the following lemma, there exists one and only one node $c_j$ in $l$ that has a Hamming distance of $i$ from $c_1$, which means that $c_j \in \mathcal{D}_{i|c_1}$. Therefore, every complete $i$-bit BRG cycle given index $c_1$ includes one node in $\mathcal{D}_{i|c_1}$. On the other hand, every node in $\mathcal{D}_{i|c_1}$ corresponds to a unique complete $i$-bit BRG cycle passing through $c_1$, as it results from Lemma~\ref{Thm: gray path lemma}. Therefore, all complete $i$-bit BRG cycles are considered in stage $i$ and we maximize their probabilities in this stage of the assignment.  

\begin{lemma}\label{Thm: uniqueness}
For each node $c_i$ in a complete $x$-bit BRG cycle, there exists one and only one node with the Hamming distance of $x$ from $c_i$.
\end{lemma}
\begin{proof}
A complete $x$-bit BRG cycle includes $2^x$ nodes and only $x$ bits are affected. Therefore, the only index that can exist with the Hamming distance of $x$ from $c_i$ is the one in which all $x$ Hamming bits are flipped.
\end{proof}

The assignment process in the stage $i$ of GO creates a bipartite graph, i.e., $(\mathcal{H}_1,\mathcal{H}_2,\mathcal{E}_3)$, where $\mathcal{H}_1$ and $\mathcal{H}_2$ are two set of nodes, and $\mathcal{E}_3$ represents the set of edges. In this stage, the nodes in sets $\mathcal{D}_{1|c_1}$, $\mathcal{D}_{2|c_1}$,...,$\mathcal{D}_{i-1|c_1}$ are already assigned and we aim to find the best assignment for the nodes in $\mathcal{D}_{i|c_1}$ such that $p(\mathcal{L}_i|v_r)$ is maximized. Among the remaining nodes in $\mathcal{V}$, we choose $\binom{k}{i}$ of them that have the highest probabilities, as $|\mathcal{D}_{i|c_1}|=\binom{k}{i}$, and allocate them to $\mathcal{H}_1$.

On the other hand, for each node $c_j$ in $\mathcal{D}_{i|c_1}$, we construct the unique complete $i$-bit BRG cycle including $c_j$ and $c_1$. Let us represent this cycle by $l_j$. Note that all nodes included in $l_j$ are assigned except $c_j$. The algorithm calculates the probability of the set of nodes in $l_j$ excluding $c_j$ and allocates it to a node in $\mathcal{H}_2$. Based on~(\ref{base}), this probability can be calculated as:
\begin{equation}\label{base2}
	p(l_j \diagdown \{c_j\}) = \prod_{v\in l_j \diagdown \{c_j\}} p(v),
\end{equation}
The algorithm repeats the process for all nodes in $\mathcal{D}_{i|c_1}$.
Next, the nodes in $\mathcal{H}_2$ are sorted, and the best matching is conducted between these two sets of nodes by assigning the $i^{th}$ node of $\mathcal{H}_1$ to the $i^{th}$ node of $\mathcal{H}_2$. The optimality of the matching process is proven in Lemma~\ref{Thm: matching}, and the achievement of maximal assignment in each stage is proven in Lemma~\ref{Thm: main proof}.

\begin{figure}[t]
	\subfloat[Sample grid.\label{Fig: Sample grid 2}]{%
		\includegraphics[scale=.38]{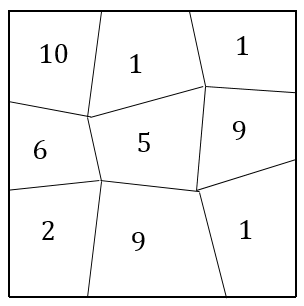}

	}
	\hfill
	\subfloat[Graph $G_1(\mathcal{C}, \mathcal{E}_1)$.\label{Fig: k-cube}]{%
		\includegraphics[scale=.38]{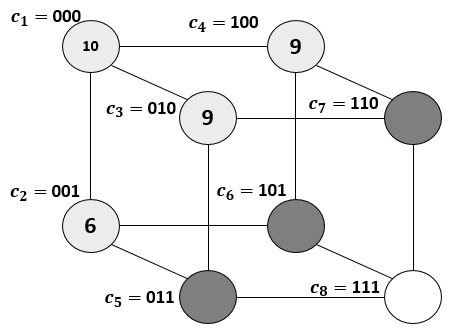}
	}
	\hfill
	\centering
	\subfloat[Graph $G_2(\mathcal{V}, \mathcal{E}_2)$.\label{Fig: Hungarian}]{%
	\includegraphics[scale=.38]{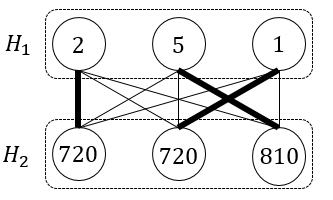}
	}
	\caption{An example of embedding graphs generated based on a sample grid.}
\end{figure}

\begin{lemma}\label{Thm: matching}
Suppose in the $i^{th}$ step of the algorithm $h_1$ to $h_{\binom{k}{i}}$ are the members of $\mathcal{H}_1$ and $h_1'$ to $h_{\binom{k}{i}}'$ are the members of $\mathcal{H}_2$ such that $h_1\leq h_2 \leq ... \leq h_{\binom{k}{i}}$ and $h_1'\leq h_2' \leq ... \leq h_{\binom{k}{i}}'$. The optimal value of matching is achieved when $h_i$ is matched with $h_i'$.
\end{lemma}

\begin{proof}
Suppose that the converse is true. Hence, there exist two nodes $h_i$ and $h_k$ which are paired with $h_j'$ and $h_t'$, respectively, such that  $h_i \leq h_k$ and $h_j' \geq h_t'$. Since the current matching is maximal by swapping $h_j'$ and $h_t'$, we have 

\begin{equation}\label{Equ: cost function}
        h_ih_j'+h_kh_t' +R > h_ih_t'+h_kh_j'+R,
\end{equation}
where $R$ indicates the remaining pairing summation. Re-writing equation (\ref{Equ: cost function}) results in
\begin{equation}\label{Equ: cost function}
        (h_i - h_k)\times (h_j'-h_t') > 0.
\end{equation}
However, $h_i \leq h_k$ and $h_j' \geq h_t'$, therefore, the left hand side of the equation is always less than or equal to zero, which is a contradiction. The case for equality of equation (\ref{Equ: cost function}) is removed as swapping does not change the summation and the lemma holds.   
\end{proof}

\begin{lemma}\label{Thm: main proof}
In stage $i$, GO maximizes $p(\mathcal{L}_i|v_r)$ given the currently assigned nodes ($\mathcal{D}_{1|c_1},\, \mathcal{D}_{2|c_1}, ...,\, \mathcal{D}_{i-1|c_1}$).
\end{lemma}
\begin{proof}
We prove the lemma based on mathematical induction. \\
\textit{Base case:} For $i=1$, given that the node $v_r$ is assigned to $c_1$, we aim to prove that GO maximizes $p(\mathcal{L}_1|v_r)$. To start with, GO chooses $\binom{k}{1}$ remaining nodes of $\mathcal{V}$ for the purpose of assignment. The optimal assignment of nodes in $\mathcal{D}_{1|c_1}$ is a permutation of the chosen nodes; otherwise, they could be replaced with a node with a higher probability that would result in a higher value for $p(\mathcal{L}_1|v_r)$. Next, the algorithm generates a bipartite graph $(\mathcal{H}_1,\mathcal{H}_2,\mathcal{E}_3)$. The probability of chosen nodes are allocated to $\mathcal{H}_1$, and the nodes in $\mathcal{H}_2$ represent the probability of complete 1-bit gray cycles constructed from $c_j \in \mathcal{D}_{1|c_1}$ and the node $c_1$, excluding the probability of $c_j$ itself. Next, the optimal matching is done by assigning the $j^{th}$ maximum node in $\mathcal{H}_2$ to the $j^{th}$ maximum node in $\mathcal{H}_1$, achieving maximal $p(\mathcal{L}_1|v_r)$ given the node $c_1$.

\textit{Induction step:} Let us assume that GO has maximized the probabilities of complete x-bit BRG cycles for $x=1\textrm{ to }i-1$ in stages one to $i-1$. We prove that in stage $i$, the algorithm maximizes complete $i$-bit gray cycles, given the previously assigned nodes. 

Based on Lemma~\ref{Thm: uniqueness}, all complete $i$-bit BRG cycles are considered in stage $i$, as each such cycle includes exactly one node in $\mathcal{D}_{i|c_1}$, which has the highest Hamming distance from $c_1$. {\em GO} starts by choosing the cells with the highest probabilities and assigning them to $\mathcal{H}_1$. Same as in the base case, we know that the optimal assignment in this stage includes the chosen set of nodes. Next, the nodes in $\mathcal{H}_2$ are assigned based on finding the probability of complete $i$-bit BRG cycles for nodes in $\mathcal{D}_{i|c_1}$, excluding the nodes themselves from the probability. As the matching process is optimal match, the best permutation of nodes in $\mathcal{H}_1$ is matched to complete $i$-bit BRG cycles.\qed
\end{proof}

\section{Scaling Up Gray Optimizer}\label{Sec: Scaling Up Gray Optimizer} 

The GO algorithm can lead to significant improvements in the processing of HVE operations; however, there are two major drawbacks once the algorithm is applied to grids with high granularities. {\em (i)} The complexity of the algorithm creates a processing time bottleneck for its application in HVE; {\em (ii)} The calculation of probabilities for large complete BRG cycles may result in numerical inaccuracies. To make {\em GO} applicable to grids with higher levels of granularity, we propose two variations.

The first proposed algorithm, called {\em Multiple Seed Gray Optimizer (MSGO)} (Section~\ref{Sec: Multiple Seed Gray Optimizer (MSGO)}), generates non-overlapping clusters and applies {\em GO} within each one of them. The second algorithm, called Scaled Gray Optimizer (SGO) (Section~\ref{Sec: Large Grids}) takes a \textit{Breadth-First Search (BFS)}~\cite{leiserson2001introduction} approach. The performance of BFS is preferred to its counterpart \textit{Depth-First Search (DFS)} as the nodes closer to the seed have higher probabilities. Thus, it is reasonable to consider those nodes earlier in the process.

\subsection{Multiple Seed Gray Optimizer (MSGO)}\label{Sec: Multiple Seed Gray Optimizer (MSGO)}

The starting point of the {\em GO} algorithm, which we refer to as {\em seed}, was chosen as the node in $G_2$ with the maximum probability. However, the algorithm can work starting with any initial seed, then follow the assignment process for other nodes in ascending order of their Hamming distance from the seed.
Furthermore, as BRG cycles become larger, their associated probability becomes smaller. Thus, one way to reduce the complexity of {\em GO} is to run the algorithm up to a particular \textit{depth}. 
Essentially, the algorithm aims at optimizing BRG cycles up to a certain length. We enhance {\em GO} by running Algorithm~\ref{Algo: Mapping Algorithm} with multiple seeds, and also by limiting the depth of the assignment. 
\begin{defn}\label{Depth}
	$\textit{Depth: }$ For a given seed $c_j$, the GO algorithm is said to run with a depth of $i$ if it only considers the assignment of nodes in $\mathcal{D}_{1|c_j},\, \mathcal{D}_{2|c_j}, ...,\, \mathcal{D}_{i|c_j}$.
\end{defn}

The pseudocode of the proposed approach is presented in Algorithm~\ref{Alg: Large Grid Mapping Algorithm (MSGO)}. The algorithm starts by assigning the node with the highest probability in $G_2$ to the origin of $G_1$ or a random index. However, instead of running GO with respect to this index for all depths from one to $k$, MSGO runs GO with the specified depth as input. The algorithm completes the process of assignment for a cluster of indices in $G_1$. MSGO then chooses a random index of $G_1$ among the remaining indices and assigns it to the node in $G_2$ with maximum probability among remaining nodes. Similarly, this index is used as a seed for GO with the specified depth and generates a new cluster. The cluster-based approach continues until all nodes are assigned to an index. The algorithm supports variable cluster sizes based on the underlying application.

\begin{algorithm}[tbh]
	\DontPrintSemicolon 
	\SetKwInOut{Input}{Input}\SetKwInOut{Output}{output}
	\Input{$G_1$; $G_2$; \textit{depth}}
	Sort nodes in $G_2$ based on probabilities\\ 
    Select a random index on $G_1$ which is not currently assigned\\
    Assign the index with the node that has the maximum probability in $G_2$\\
    Apply Algorithm~\ref{Algo: Mapping Algorithm} on the selected index with the
    specified depth\\
    \textbf{Repeat} lines 2-4 \textbf{until} all indices are assigned\\ 
	\caption{Multiple Seed Gray Optimizer~(MSGO).}
	\label{Alg: Large Grid Mapping Algorithm (MSGO)}
\end{algorithm}

\begin{comment}
\begin{exmp}\label{Example: 3}
Consider the sample map provided in Fig.~\ref{Fig: Sample grid} and two banks $A$ and $B$ requested for location-based alert notification service with three and four branches, respectively, located on the map. Suppose that the headquarters of the banks have the highest probability of becoming an alert zone among the branches of each bank. In this scenario, MSGO proceeds by executing GO on three different clusters consisting of branches of bank $A$, branches of bank $B$, and the remaining cells of the grid. In the two clusters corresponding to the banks, headquarters would be the seeds used as input to GO, and for the third cluster, the cell with the highest probability is chosen as the seed. The depth of GO in each cluster is adjusted based on the corresponding number of cells.  
\end{exmp}
\end{comment}

The MSGO algorithm provides a robust solution for grids with higher granularity. The algorithm no longer suffers the drawbacks of GO when the grid size grows, such as numerical inaccuracies in the calculation of the probability of large cycles. The complexity of the algorithm depends on the depth chosen as input, and in low depths, it can be implemented in $\mathcal{O}(n(\log_2n))$. MSGO can significantly reduce the number of operations required for the implementation of HVE in location-based alert systems, and therefore, making it a practical solution for preserving the privacy of users in location-based alert systems. 

\subsection{Scaled Gray Optimizer (SGO)}\label{Sec: Large Grids}

SGO considers overlapping clusters and necessitates that all nodes act as seed during the assignment process. The pseudocode of the proposed approach is presented in Algorithm~\ref{Algo: Large Grid Mapping Algorithm}. SGO starts by assigning the node with the highest probability to an index on $G_1$. However, instead of assigning indices with all depths from one to $k$ with respect to index $c_1$, the SGO algorithm runs GO with the depth of one. 
Next, SGO sorts the indices in $\mathcal{D}_{1|c_1}$ based on their assigned probabilities in descending order and runs GO with the depth of one on each index. 
Once the algorithm is applied on all the indices in $\mathcal{D}_{1|c_1}$, the process repeats for indices in $\mathcal{D}_{2|c_1}$, $\mathcal{D}_{3|c_1}$, ..., etc. The algorithm continues until all indices are assigned to a node.

\begin{algorithm}[tbh]
	\DontPrintSemicolon 
	\SetKwInOut{Input}{Input}\SetKwInOut{Output}{output}
	\Input{$G_1$; $G_2$}
	Sort nodes in $G_2$ based on probabilities\\
    Assign $v_r\in G_2$ with the highest probability to the origin of $G_1$, i.e., $c_1$\\
    Apply Algorithm~\ref{Algo: Mapping Algorithm} on $c_1$ with the depth of one\\
	\For {$i\, \textrm{in}\, [1:k]$} {
	Sort $\mathcal{D}_{i|c_1}$ in descending order of probabilities assigned to its indices\\
	\For {$c_j$ in $\mathcal{D}_{i|c_1}$} {
	Apply Algorithm~\ref{Algo: Mapping Algorithm} on $c_j$ with the depth of one}
	}
	\caption{Scaled Gray Optimizer (SGO).}
	\label{Algo: Large Grid Mapping Algorithm}
\end{algorithm}

\begin{comment}
\begin{exmp}\label{ex3t}
Consider a map of a vast rural area in which there is a likelihood of bush fire or security breaches for the residents. Therefore, the alert based notification system is implemented by a service provider to notify the farmers on their request. Due to a large number of cells included in the map, GO and MSGO require a comparably high computation overhead. Hence, the SGO algorithm could be used to improve the time complexity. Starting with the cell that has the highest probability, SGO executes GO with the minimum depth, i.e., one. Then, in a breadth-first-search manner, the algorithm moves to the first neighbors of that node, taking them as seeds of GO, prioritizing the nodes with higher assigned probability. The process continues until all the nodes of the grid are assigned.
\end{exmp}
\end{comment}

%The SGO algorithm can significantly reduce the number of operations required for the implementation of HVE in location-based alert systems, and therefore, making it a practical solution for preserving the privacy of users in location-based alert systems. 

\subsection{Complexity Analysis}\label{Sec: Complexity Analysis} 

The key computation overhead of the GO algorithm is in the calculation of probability of BRG cycles. Let the function $T(.)$ return the computational complexity. In the $i^{th}$ step of the algorithm, the nodes with the hamming distance of $i$ from $c_1$ are assigned to an index on the $k$-cube, i.e., $\mathcal{D}_{i|c_1}$. The number of nodes in $\mathcal{D}_{i|c_1}$ is $\binom{log_2(n)}{i}$. For each one of such nodes the complete BRG path is calculated which requires the multiplication of $2^i-1$ probabilities. Therefore, the assignment process for the nodes in $\mathcal{D}_{i|c_1}$ requires 
\begin{equation}
    T(\mathcal{D}_{i|c_1}) = (2^i-2)\times \binom{log_2(n)}{i}
\end{equation}
operations. Hence, the total number of operations required for the algorithm is

%\begin{equation}
\begin{align} 
   \bigcup_i T( \mathcal{D}_{i|c_1}) & =  \sum_{i=1}^{log_2(n)}(2^i-2)\times \binom{log_2(n)}{i} \\
    & = \sum_{i=1}^{log_2(n)} 2^i\binom{log_2(n)}{i} -2\sum_{i=1}^{log_2(n)}\binom{log_2(n)}{i}\label{Equ: complexity}. 
\end{align}
From the binomial theorem,
\begin{equation}
    \sum_{i=1}^{log_2(n)} 2^i\binom{log_2(n)}{i} = (1+2)^{log_2(n)}-1 = 3^{log_2(n)}-1,
\end{equation}
and 
\begin{equation}
    \sum_{i=1}^{log_2(n)}\binom{log_2(n)}{i} = 2^{log_2(n)}-1.
\end{equation}
Therefore, Eq.~\eqref{Equ: complexity} can be written as

\begin{equation}
   \bigcup_i T( \mathcal{D}_{i|c_1}) = n^{log_23} - 2n+1.
\end{equation}

In addition to the above operations, there exists an initial sorting of the probabilities that can be implemented based on merge sort with the complexity of $\mathcal{O}(n(\log_2n))$, and a sorting process in each stage for the nodes in $\mathcal{H}_2$. For the latter, the complexity can be written as 
\begin{align}
   &\sum_{i=1}^{log_2(n)} \mathcal{O}(\binom{log_2(n)}{i}\times \log_2(\binom{log_2(n)}{i})) = \\
   & \sum_{i=1}^{log_2(n)} \mathcal{O}(  
   \log_2\binom{log_2(n)}{i}^{\binom{log_2(n)}{i}})\leq \\
   & \sum_{i=1}^{log_2(n)} \mathcal{O}(  
   \log_2 n^{\binom{log_2(n)}{i}})\leq \\
   &  \mathcal{O}(  
   \log_2 n^{\sum_{i=1}^{log_2(n)}\binom{log_2(n)}{i}}) = \mathcal{O}(n(\log_2n))
\end{align}

Therefore, accounting for sorting, the closed-form expression for the total complexity is $\mathcal{O}(2n(\log_2n)) + n^{log_23} - 2n+1$.

The MSGO algorithm is based on executing the GO algorithm with shorter depths in a cluster based approach. Suppose that the depth is set to $r$ where $r\leq log_2n$. Running the algorithm in each cluster with similar logic as the GO requires the following number of operations.

\begin{equation}
    \sum_{i=1}^{r} (2^i-2)\binom{log_2(n)}{i}.
\end{equation}
On the other hand, the total number of clusters is approximately
\begin{equation}
    \# clusters \approx n/ \sum_{i=1}^{r} \binom{log_2(n)}{i}.
\end{equation}

Therefore, the total complexity considering the initial sorting algorithm is calculated as

\begin{equation}\label{Equ: approximate}
\begin{split}
    \mathcal{O}(n(\log_2n))&+\\ (\sum_{i=1}^{r} &(2^i-2)\binom{log_2(n)}{i})\times (n/ \sum_{i=1}^{r} \binom{log_2(n)}{i}).
\end{split}
\end{equation}
Defining the binary entropy function as
\begin{equation}
    H_2(x) = x \log_2(\dfrac{1}{x})+ (1-x)log(\dfrac{1}{1-x}),
\end{equation}
the following approximation can be used for deriving closed-form expression for various cluster sizes in Eq.~\eqref{Equ: approximate} 
\begin{equation}
    \binom{n}{i} 	\simeq  2^{nH_2(r/n)}.
\end{equation}

Lastly, the SGO algorithm executes the GO algorithm with the depth of one and has the computational complexity of $\mathcal{O}(n(\log_2n))$. The low computational complexity of SGO makes it a suitable option for the encoding of grids with higher levels of granularity. 

\section{Supporting Dynamic Alert Zones}\label{Advance Modeling of Alert Zones}

So far, we considered the case of static alert zones, and we optimized the data encoding and token generation under this scenario. However, in practice, alert zones vary over time. Whether an alert corresponds to a natural phenomenon (e.g., gas leak) or a human activity (e.g., COVID carrier movement), alert zones exhibit spatio-temporal patterns that must be accounted for in order to obtain fast performance.

We maintain the grid-based partition of the spatial domain used for the static case, and we denote by {\em state} of the grid the set of all alert cells at a given time. The occurrence probability of a state can be modeled analytically and used as a basis for grid encoding. The higher the statistical model accuracy, the more precise the encoding becomes, reducing HVE operations overhead. 
%So far, we have assumed that there exists no dependence in time and space among cells whereas in practical scenarios, cells might be correlated, and the likelihood of alert cells change once an event is triggered. 
Next, we build a comprehensive statistical model to characterize alert zone evolution in space and time.  

\begin{defn}(State Space).\label{Definition: State Space}
    For a given grid $$\mathcal{V}=\{ v_1,\,,v_2,..., v_{n} \},$$ let $X$ be a random variable defined on all possible subsets of the cells. The state space of $X$ is defined as the power set $\mathcal{S}_n=\{\textbf{\underline{1}}, \textbf{\underline{2}},...,\textbf{\underline{$2^n$}}\}$.
\end{defn}

The {\em cardinality} of a state $\textbf{\underline{i}}$ represents the number of cells included in the state and is denoted by $|\textbf{\underline{i}}|$. The set of all states with the cardinality of $j$ are denoted by $\mathcal{S}^{|j|}_n$. Note that, the notation is not concerned with a precise order of states. For example, a grid with two cells $\{ v_1, v_2\}$ leads to the state space of 

\noindent
$\mathcal{S}_2 = \{\{\emptyset\},\{v_1\},\{v_2\}, \{v_1,v_2\}\}$, which is depicted by $\mathcal{S}_2=\{\textbf{\underline{1}}, \textbf{\underline{2}},\textbf{\underline{3}},\textbf{\underline{4}}\}$; however, the order of states is not captured by the notation. Two examples of such an assignment can be $\{\textbf{\underline{1}} = \{\emptyset\}, \textbf{\underline{2}} = \{v_1\},\textbf{\underline{3}} = \{v_2\},\textbf{\underline{4}} = \{v_1,v_2\} \}$,  and $\{\textbf{\underline{1}} = \{   \{v_1,v_2\}  \}, \textbf{\underline{2}} = \{v_2\},\textbf{\underline{3}} = \{v_1\},\textbf{\underline{4}} = \emptyset \}$. We provide more details on the construction of the state space and ordering in Section~\ref{subsection: Recursive Construction}. 

\begin{comment}
\begin{figure*}[h!]
	\subfloat[x = 0.01]{%
	\includegraphics[scale=.6]{Figures/ThrreeDx01.png}
	}
	\hfill
	\subfloat[x = 0.1]{%
	\includegraphics[scale=.6]{Figures/ThrreeDx1.png}
	}
	\hfill
	\centering
	\subfloat[x = 0.5]{%
	\includegraphics[ scale=.6]{Figures/ThrreeDx05.png}
	}
	\caption{Sigmoid function with variable a, b, and x.}
	\label{dist}
\end{figure*}	
\end{comment}

Let $X_0, X_1,...,X_i,...$ denote the sequence of random variables modeling the occurrence of alert zones. The set of possible values for $X_i$ is the state space of the grid, and the index $i$ denotes the evolution of the process in time. The probability of $X_i$ being in a particular state $\textbf{\underline{j}}$ is denoted as $p(X_i = \textbf{\underline{j}})$. The probability of a cell becoming part of an alert zone depends on underlying phenomena properties, existing correlations among cells, and the history of alert zones on the map. Moreover, probabilities do not remain constant over time. We identify several distinct scenarios, and we create a statistical model for each: (i) the states are independent in both space and time; (ii) the states are independent in space, but dependent in time (i.e., temporal causality); (iii) the states are independent in time but exhibit space correlation (i.e., spatial causality); and (iv) the states are dependent in both time and space. The first case corresponds to the static case introduced in the previous sections; the last case is the most general one, whereas cases (ii) and (iii) are special cases of (iv). Each case may be relevant under different types of applications and data domains. Next, we investigate in details each of the cases, and propose a data encoding and token generation technique for each. Our goal is to obtain an accurate representation of how the probabilities $X_i$ are distributed over the state space.

\subsection{Independence in Time and Space}\label{Sect: Independence in Time and Space}

Having the independence assumption in space and time greatly simplifies the problem formulation as the sequence of random variables $X_0, X_1,...,X_i,...$ become a sequence of independent and identically distributed (iid) random variables defined over the state space. Such modeling indicates that the random variables $X_1$ to $X_i$ provide no information about the random variable $X_{i+1}$. Therefore, the {\em probability mass function (PMF}) of $X_i$ depends on the probabilities of individual cells. For a given $X_i$, the probability of cell $v_i\in \mathcal{V}$ becoming part of the alert zone is denoted by $p(v_i)$, and corresponds to a value between zero and one. The mutual probability of a subset of cells $\mathcal{L} = \{ v_1,\,,v_2,..., v_{i} \}$ being in an alert zone can be calculated as 
\begin{equation}
	p(\mathcal{L}) = \prod_{j=1}^{i} p(v_j).
\end{equation}
The calculation of mutual probabilities is the direct result of the independence assumption, which indicates that there are no correlations between cells.

%To model the probability values for the purpose of experiments in this paper, we use the sigmoid function $sig(x)= 1/(1+\exp^{-b(x-a)})$, where $a$ and $b$ are parameters controlling the function shape. The output value is a score between zero and one representing the likelihood of $v_i\in \mathcal{V}$ becoming an alert zone, where a higher value corresponds to a higher likelihood. The sigmoid function is a frequent model used in machine learning, and we choose it because we expect that, in practice, the probability of individual cells becoming part of an alert zone can be computed using such a model built on a regions' map of features (e.g., type of terrain, building designation, point-of-sale information, etc.). Parameter $a$ of the sigmoid controls the \textit{inflection} point of the curve, whereas $b$ controls the gradient. Fig.~\ref{dist} plots the logistic function for several different values of $a$ and $b$ which we use in our evaluation. As can be observed from the figure, changing the parameters' values provides flexibility to reach a wide range of possible PMFs.

\subsection{Independence in Space, Dependence in Time}\label{Sec: Independence in Space, Dependence in Time}

%As before, let $X_0, X_1,...,X_i,...$ denote the sequence of random variables modeling the occurrence of alert zones over the state space. 
In this case, the grid state no longer consists of iid random variables following the same PMFs. The probability of state $\textbf{\underline{i}}$ at time $j$ is no longer assumed to be equal to the probability of being in state $\textbf{\underline{i}}$ at a different time $k$, i.e.,  $p(X_j = \textbf{\underline{i}}) \neq p(X_k = \textbf{\underline{i}})$. Our objective is to determine whether the system reaches a steady state in which the probabilities no longer change significantly over time. We model the evolution of alert cells over time using Markov chains. 
%Note that, the alert zones are often discrete events on the map and have a slow evolution in time. Therefore, an order one Markov chain will result in a better approximation of the real world events. 
We assume that alert zones evolve incrementally by addition or removal of a single cell at a time (this can always be achieved by properly choosing the time granularity).

%Markov chains were first introduced to show that the law of large numbers can be applied on random variables which are not independent. 
%A Markov chain is a gateway between completely independent $X_i$s which provide no information about each other and $X_i$s which have extreme interactions that even the basic properties cannot be computed. 
The proposed model is represented in Fig~\ref{fig:markov model}. States $\textbf{\underline{i}}$ and $\textbf{\underline{j}}$ are connected if and only if the difference between their cardinality is one, $|\textbf{\underline{i}}-\textbf{\underline{j}}|=1$. The only exception is the state including all cells (if all cells are within the alert zone, then all have the same status). The model assumes that each state depends only on the previous state, and therefore, it follows Markov chain properties, i.e., for all $k\geq 0$,

	\begin{equation}
	\begin{split}
	    p(X_{k+1}=\textbf{\underline{j}}|X_k =\textbf{\underline{i}},X_{k-1}=\textbf{\underline{i}}_{k-1}&,...,X_0 =\textbf{\underline{i}}_0 )\\&
	    =  p(X_{k+1}=\textbf{\underline{j}}|X_k =\textbf{\underline{i}}).
	\end{split}
	\end{equation}
The forward propagation to a state with a higher cardinality indicates the addition of an alert cell, whereas forward propagation to a state with a lower cardinality indicates the removal of an alert cell. 

The value of $p(X_{k+1}=\textbf{\underline{j}}|X_k =\textbf{\underline{i}})$ is called the {\em transition probability} from state $\textbf{\underline{i}}$ to state $\textbf{\underline{j}}$ and we implicitly make the assumption that the transition probabilities are homogeneous over time. We are interested in understanding what the likelihood of being in a state is starting from any other state, and whether the chain reaches a {\em stationary distribution} in which the probabilities of individual states do not change over time. 

First, we review three properties of the proposed Markov chain:

\begin{prop}\label{recurrency}
All states in the proposed model are recurrent. Therefore, starting from any state of the chain, it is possible to reach any other state, eventually.
\end{prop}

\begin{prop}\label{irreducable}
The proposed Markov chain is irreducible, as for any two states $\textbf{\underline{i}}$ and \textbf{\underline{j}}, it is possible to reach one from the other in a finite number of steps. 
\end{prop}

\begin{prop}\label{periodic}
The proposed Markov chain for modeling alert zones is aperiodic, as the period of states is equal to one. 
\end{prop}

The above properties help to characterize the long-term behaviour of the Markov chain. If after a certain period of time the {\em transition matrix} of the chain reaches a stationary distribution, it enables us to know the probability of each state in the state space. The state {\em transition matrix} is defined as follows:

\begin{figure}[t]
\centering
\includegraphics[scale=.5]{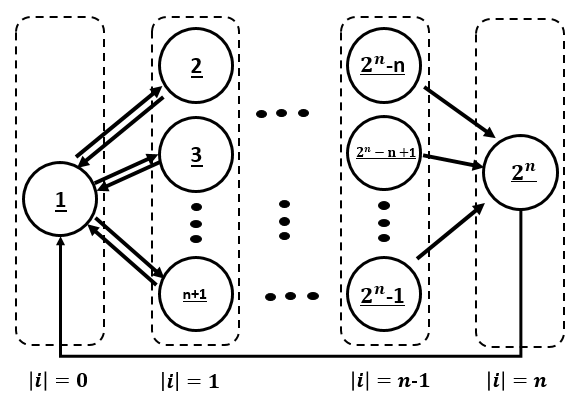}
\hspace{1em}
\centering
\caption{Proposed Markov model for alert zone evolution.}
\label{fig:markov model}
\end{figure}

\begin{defn}(Transition matrix).\label{Markov Chain}
    For a Markov chain $X_0, X_1,...,X_i,...$ with a state space $\mathcal{S}_n=\{\textbf{\underline{1}}, \textbf{\underline{2}},...,\textbf{\underline{$2^n$}}\}$, let $q_{ij} =p(X_{k+1}=\textbf{\underline{j}}|X_k =\textbf{\underline{i}})$ be the transition probability from state $\textbf{\underline{i}}$ to state $\textbf{\underline{j}}$. The $2^n \times 2^n$ matrix $Q_n = (q_{ij})$ is called the transition matrix of the chain. The value of $q_{ij}$ for $i<2^n$ is defined as $p(v)$, where $v$ is the alert cell which exists in state $\textbf{\underline{i}}$ (row) and does not exist in state $\textbf{\underline{j}}$ (column). 
\end{defn}

Recall that two states are connected if and only if their cardinality differs by one. The last row of the matrix represents the only outgoing directed edge from the state with the cardinality of $n$ to the state with cardinality of zero. Thus, the first element of the last row is one ($q_{2^n1}=1$) and all its other elements are zero. Such a row ensures the aperiodicity of the chain.

It can be inferred that the $i^{th}$ row of the transition matrix corresponds to outgoing edges from the state $\textbf{\underline{i}}$ of the Markov chain. Therefore, in order for the matrix to maintain the Markovian properties, the values in each row should sum up to one, which is indeed the case for the proposed transition matrix. This property is termed as {\em Markovian matrix property}. Let a row vector $\boldsymbol{t} = [t_1,t_2,...,t_{2^n}]$ be the PMF of $X_0$, where $t_i=p(X_0 = \textbf{\underline{i}})$. Then, based on the properties of Markovian chains, the marginal distribution of $X_m$ is given by the $j^{th}$ component of $\boldsymbol{t}Q_n^m$, i.e., $p(X_n= \textbf{\underline{j}})$. The marginal distribution indicates that given a initial state $\textbf{\underline{i}}$, the probability of being in state \textbf{\underline{j}} after $m$ transitions is the $j^{th}$ component of the vector $\boldsymbol{t}Q_n^m$. We are interested in the long run behaviour of the system and to understand if the proposed model will reach a {\em stationary distribution}.

\begin{defn}(Stationary distribution).\label{Dfn: stationary distribution}
    Given a Markov chain with the transition matrix $Q_n$, a row vector $\boldsymbol{s} = [s_1,...,s_{2^n}]$, such that $s_i\geq 0$ and $\sum_i s_i = 1 $, is a stationary distribution if \begin{equation}\label{equ: stationary distribution}
        \boldsymbol{s}Q_n=\boldsymbol{s}
    \end{equation} 
\end{defn}
We elaborate further on the meaning of the vector $\boldsymbol{s}$. Suppose that the $i^{th}$ element of the vector corresponds to the state $\textbf{\underline{i}}$. If the proposed Markov chain reaches a stationary distribution, this value represents probability $p(X_n= \textbf{\underline{i}})$ for any $n$ after reaching the stationary distribution. Thus, the importance of each state is revealed by its corresponding value in $\boldsymbol{s}$. 

\begin{exmp}\label{example: transition matrix}
Consider a map with two cells $v_1$ and $v_2$, where $p(v_1) = 0.2$ and $p(v_2) = 0.8$. The state space includes four states $\{\{\emptyset\},\{v_1\}\,\{v_2\}\,\{v_1,v_2\}\}$, and the transition matrix is calculated as
\begin{equation}\nonumber
Q_2=
\begin{bmatrix}
0 & 0.2 & 0.8 & 0\\
0.2 & 0 & 0 & 0.8\\
0.8 & 0 & 0 & 0.2\\
1 & 0 & 0 & 0
\end{bmatrix},
\end{equation}

Solving Eq.~\eqref{equ: stationary distribution} for the matrix $Q_2$ results in the eigen vector $\boldsymbol{s} = [0.4310,\, 0.0862,\, 0.3448 ,\, 0.1379]$. Hence, the probability of states are $p(\{\emptyset\}) = 0.4310$, $p(\{v_1\}) = 0.0862$, $p(\{v_2\}) = 0.3448$, $p(\{v_1,\, v_2\}) = 0.1379$. 
\end{exmp}

There are three important questions to be answered about the stationary distribution: (a) {\em does it exist?} (b) {
\em is it unique?} and (c) {\em does the Markov chain converge to the stationary distribution?} 
The stationary distribution is the left eigenvector of the transition matrix corresponding to the eigenvector of one as shown by Eq.~\eqref{equ: stationary distribution}. The existence and uniqueness of a stationary distribution for the proposed Markov model is proven in the following theorem.

\begin{thm}
There exists a unique stationary distribution for the proposed Markov chain to model alert zones. 
\end{thm}
\begin{proof}
According to~\cite{blitzstein2019introduction}, a stationary distribution exists for any finite-state Markov chain, and if the chain is irreducible, the solution is unique. Based on property~\ref{irreducable}, there exists a unique stationary distribution for the  model. Later in Section~\ref{subsection: Recursive Construction}, we present the recursive construction of matrix $Q_n$ and show that the cardinality of the null space of the matrix $\boldsymbol{s}(Q_n-I)$ is one. 
\end{proof}

The above theorem shows that there exists a unique stationary distribution for the proposed Markov model regardless of the initial probabilities of the cells; however, to reach the stationary distribution, the chain needs to be aperiodic as well as irreducible. Based on Property~\ref{periodic}, the proposed model is aperiodic. However, particular initial probabilities, including zero values, can result in periodic chains. To address this problem, we adopt a similar approach as the PageRank algorithm~\cite{sarma2013fast}, used to rank the relevance of webpages. Suppose that before moving to a new state on the chain, a coin is tossed with probability $\alpha$ of heads. If the result of the coin toss is heads, the state evolves using the transition matrix $Q$; otherwise, the system jumps to a state in a uniformly random distribution. The resulting transition matrix is represented as:

\begin{equation}
    O_n = \alpha Q + (1-\alpha)\dfrac{J_n}{2^n},
\end{equation}
where $J_n$ is a $2^n \times 2^n$ matrix of all ones. The recommended value~\cite{sarma2013fast} of $\alpha$ is $0.85$. It can be observed that all elements of $O_n$ are positive, and therefore, the aperiodicity of the chain is guaranteed. Hence, solving Eq.~\eqref{equ: stationary distribution} for $O_n$ has a solution leading to a stationary distribution ($\boldsymbol{s}$) as well as converging to the stationary distribution. Similarly, the $i^{th}$ element of the vector $\boldsymbol{s}$ for the new transition matrix $O_n$ indicates the significance of state $\textbf{\underline{i}}$, as it represents $p(X_m= \textbf{\underline{i}})$ for any large value of $m$. In the following, we consider that the transition matrix is aperiodic, and we use the matrix $Q_n$ as our reference.

\begin{exmp}\label{example: Google transition matrix}
Going back to Example~\ref{example: transition matrix}, the transition matrix is derived as 
\begin{equation}\nonumber
O_2=
\begin{bmatrix}
0.0375 & 0.2075 & 0.7175 & 0.0375\\
0.2075 & 0.0375 & 0.0375 & 0.7175\\
0.7175 & 0.0375 & 0.0375 & 0.2075\\
0.8875 & 0.0375 & 0.0375 & 0.0375
\end{bmatrix}.
\end{equation}

Solving Eq.~\eqref{equ: stationary distribution} for the matrix $O$ results in the eigenvector of $\boldsymbol{s} = [0.4111,\,  0.1074,\, 0.3171,\, 0.1644]$. Hence, the new probability of states are $p(\{\emptyset\}) = 0.4111$, $p(\{v_1\}) = 0.1074$, $p(\{v_2\}) =  0.3171$, $p(\{v_1,\, v_2\}) = 0.1644$. One can check the convergence by choosing a large enough value of $m$ and calculating $tO_2^m$ starting by an arbitrary PMF on $t$. As an example, if the vector $\boldsymbol{t} = [0.25, 0.25,0.25,0.25]$, then $\boldsymbol{t}O_2^{50}$ will result in $$[0.4111,\,  0.1074,\, 0.3171,\, 0.1644]$$ which is the stationary distribution vector~$\boldsymbol{s}$. 
\end{exmp}

\subsection{Dependence in both space and time}

%Having no correlation assumptions in space and time for the sequence of random variables $X_0, X_1,...,X_i,...$ makes it difficult to answer if there exists a steady-state for the system. Even finding a way to present the correlation of one alerted cell with the others is a challenging task. 
In this section, we study how to capture correlation among alert cells over time by incorporating spatial distance between cells within the Markov model. We embed spatial correlations in the transition matrix while maintaining Markovian properties, and thus the long-term behaviour of the model can be better defined. We use as starting point the proposed model from Fig.~\ref{fig:markov model}.

Consider a grid with two cells $\{ v_1, v_2\}$ and the state space of $\mathcal{S}_2 = \{\textbf{\underline{1}} = \{\emptyset\}, \textbf{\underline{2}} = \{v_1\},\textbf{\underline{3}} = \{v_2\},\textbf{\underline{4}} = \{v_1,v_2\} \}$. The matrix $Q_2$ is derived as

\begin{equation}\nonumber
Q_2=
\begin{bmatrix}
0 & p(\{v_1\}) & p(\{v_2\})  & 0\\
p(\{v_1\}) & 0 & 0 & p(\{v_2\})\\
p(\{v_2\}) & 0 & 0 & p(\{v_1\})\\
1 & 0 & 0 & 0
\end{bmatrix}.
\end{equation}
Investigating the transition matrix closely, one can see the impact of independence between cells in the matrix. Consider the entry $Q_2(2,4)$ as an example. This entry indicates the probability of going from state $ \textbf{\underline{2}} = \{v_1\}$ to state $\textbf{\underline{4}} = \{v_1,v_2\}$ is $p(\{v_2\})$. In other words, the transition captures the fact that the existence of another alert zone cell $v_1$ did not impact the cell $v_2$ (i.e., spatial independence between cells). More formally, from the {\em Bayes rule}:

\begin{equation}
    p(\{v_1,v_2\}) = p(v_2| v_1)p(v_1) \rightarrow p(\{v_1,v_2\}) = p(v_2)p(v_1), 
\end{equation}
given that 
\begin{equation}
    p(v_2| v_1) = p(v_2).  
\end{equation}
In Section~\ref{Sec: Independence in Space, Dependence in Time}, we assumed independence between states. To address this issue, we propose the following method to capture the correlations between states without eliminating the Markov property of the matrix $Q_n$. The main idea behind the approach is that cells that are in close proximity to the alert zone are more likely to become part of the zone in the future.

Let $X_0, X_1,...,X_i,...$ be an order one Markov sequence of random variables modeling the occurrence of the alert zones, where $X_i$'s are defined over the state space of the grid. Without loss of generality assume that the $j'^{th}$ row of the matrix $Q$ corresponds to the state $\{v_1,v_2,...,v_j \}$. Based on the proposed Markov model in Fig.~\ref{fig:markov model}, it is known that this state can evolve by the addition or removal of a single alert cell. Therefore, there exist $n$ non-zero elements in each row of the matrix. For all $v_k \in \mathcal{V}$, we calculate the probability of its removal or addition as:

\begin{equation*}
\begin{aligned}
\textrm{If}\, v_k \notin &\{v_1,v_2,...,v_j \}\, \textrm{then};\\
& p(\{v_k\}\cup \{v_1,v_2,...,v_j \}) = p(v_k)/(d(v_k,c))\times \beta\\
\textrm{If}\, v_k \in &\{v_1,v_2,...,v_j \}\, \textrm{then};\\
& p( \{v_1,v_2,...,v_j \} -\{v_k\} ) = p(v_k)/(d(v_k,c))\times \beta,
\end{aligned}
\end{equation*}
where the function $d(.)$ returns the Euclidean distance between two points, $\beta$ is a normalization factor over the entire row, and the point $c$ is the centre point of  $\{v_1,v_2,...,v_j \}$, calculated as 
\begin{equation}
    c = (\sum_{i=1}^j v_j)/j.
\end{equation}

Note that, in all above calculations, each cell's center point is used as its representative. The intuition behind the approach is that the correlation between cells becomes smaller as we go further away from the alert zone. The only special case is when there exists a single-cell alert zone, and we seek the probability of its removal. In this case $d(v_k,c)$ becomes close to zero and $p(v_k)/(d(v_k,c))$ tends to go to infinity. As there exist no other alert zone cell for this case, we consider this probability as $p(v_k)$ instead of $p(v_k)/(d(v_k,c))$ to avoid inaccuracies. As an example, consider a grid with three cells $\{ v_1, v_2, v_3\}$ and the average point $c$. Suppose that the $j^{th}$ row of the matrix $Q_3$ corresponds to the state $\{ v_1, v_2\}$. In this row, there exist three nonzero elements: 

\begin{equation}
\begin{aligned}
& p(\{ v_1\})= p(v_2)/(d(v_2,c) ) \times \beta\\
&p(\{ v_2\}) = p(v_1)/(d(v_1,c)\times \beta\\
&p(\{ v_1, v_2, v_3\}) = p(v_3/(d(v_3,c))\times \beta, \textrm{where}\\
&\beta = 1/(p(v_2)/(d(v_2,c) ) + p(v_1)/d(v_1,c) + p(v_3/d(v_3,c)))
\end{aligned}
\end{equation}

The proposed method satisfies the Markovian matrix property. Hence, it can be used as part of the Markov model in Section~\ref{Sec: Independence in Space, Dependence in Time} to capture the long-term behavior of the system.

\subsection{Recursive Construction and Monte Carlo Sampling}\label{subsection: Recursive Construction}

Finding the eigenvector of matrix $Q_n$ corresponding to eigenvalue one is necessary to determine the probability of being in a particular state at a given time $p(X_n= \textbf{\underline{i}})$. The eigenvector provides valuable information that enables us to prioritize more likely states in the grid encoding process. However, there are two important issues with its calculation: {
em (i)} The matrix $Q_n$ has dimensions of $2^{n}\times 2^{n}$. Even considering a small grid with 100 cells, it requires an extremely large storage capacity. {\em (ii)} The calculation of the eigenvector for such a large matrix is expensive, with $\mathcal{O}(n^3)$~\cite{press1988numerical} complexity. For example, based on Householder transformations, eigenvalues and eigenvectors can be calculated with complexity $\mathcal{O}(n^2) + 4n^3/3$. To address the high computational overhead, we approximate the stationary distribution based on random walks on the Markov model.

We start by explaining the recursive construction of the matrix $Q_n$. The rows and columns of the matrix depend on the order in which states are chosen. We propose to construct the states of the $n+1$ cells, $v_1$ to $v_{n+1}$, from the grid with $n$ cells, $v_1$ to $v_n$ as follows:
\begin{align}
    &\mathcal{S}_2 = \{\{\emptyset\},\{v_1\},\{v_2\}, \{v_1,v_2\}\},\\
    &\mathcal{S}_{n+1}  = \{\mathcal{S}_n ,\, \mathcal{S}_{n}\bigcup v_{n+1}\}.
\end{align}
For instance, $\mathcal{S}_3$ is constructed as 
\begin{equation}
\begin{split}
    \mathcal{S}_3 = \{\{\emptyset\},\{v_1\},\{v_2\}&,\{v_1,v_2\},\\
    &\{v_3\},\{v_1,v_3\},\{v_2,v_3\}, \{v_1,v_2,v_3\}\},
\end{split}
\end{equation}
The matrix $Q_{n+1}$ can be constructed recursively as 
\begin{equation}
    Q_n=
    \begin{bmatrix}
    W_{n-1} & p(v_n)I_{2^{n-1}} \\
    p(v_n)I_{2^{n-1}}& W_{n-1}
    \end{bmatrix}
    -W_n(2^n,:) + K_{2^n},
\end{equation}
where $I_{2^{n}}$ is the identity matrix and $K_{2^n}$ is an all-zero $2^n\times 2^n$ matrix except for element $K_{2^n}(2^n,0) = 1$, and 
\begin{equation}
    W_n=
    \begin{bmatrix}
    W_{n-1} & p(v_n)I_{2^{n-1}} \\
    p(v_n)I_{2^{n-1}}& W_{n-1}
    \end{bmatrix}
\end{equation}

\noindent given that

\begin{equation}\nonumber
W_2=
\begin{bmatrix}
0 & p(v_1) & p(v_2) & 0\\
p(v_1) & 0 & 0 & p(v_2)\\
p(v_2) & 0 & 0 & p(v_1)\\
0 & p(v_2) & p(v_1) & 0
\end{bmatrix}.
\end{equation}

The above representation of $Q_n$ works under the spatial independence assumption, but the construction of states holds regardless of that assumption.

To tackle the high computational complexity of determining eigenvectors, we use a probabilistic approach. PageRank's approach~\cite{sarma2013fast} to this problem is incorporating the power iteration method to calculate the eigenvectors, but still incurs a high computational complexity. An alternative approach is the Monte Carlo approximation, which is widely used in literature and results in an enhanced estimation of the stationary distribution. The Monte Carlo method provides several advantages over deterministic power iteration methods such as significantly lower computation complexity, opportunities for parallel implementation, and it facilitates updating of probabilities.  

The main idea behind the Monte Carlo approximation is to start $R$ random walks on the Markov model's primary node, i.e., state $\textbf{\underline{1}}$. Each random walk terminates with the probability of $1-c$ and makes a transition to the next outgoing node with the PMF specified in the transition matrix $Q_n$. The fraction of walks ending at a state over all the random walks indicates the probability or significance of that state.  The vector of calculated probabilities for all states is the approximation of stationary distribution. The number of samples required to estimate the stationary distribution is shown to pessimistically be in the order $\mathcal{O}(n^2)$, where $n$ indicates the number of states; however, it is shown that $n$ random walks are enough to provide a reasonable approximation of a stationary distribution~\cite{avrachenkov2007monte}.

\section{Experimental Evaluation}

\begin{figure*}[!ht]

	\subfloat[a=0.75, b=10.\label{F11}]{%
	\includegraphics[scale=.53]{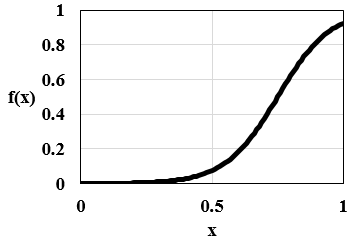}
	}
	\hfill
	\subfloat[a=0.75, b=10.\label{F12}]{%
	\includegraphics[scale=.53]{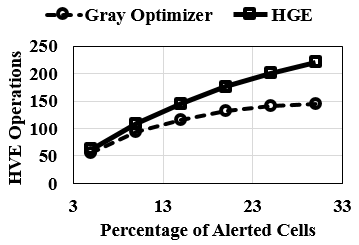}
	}
	\hfill
	\subfloat[a=0.75, b=10. \label{F13}.]{%
	\includegraphics[scale=.53]{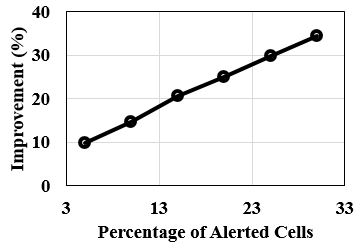}
	}
	\hfill
	\subfloat[a=0.75, b=30. \label{F21}]{%
	\includegraphics[scale=.53]{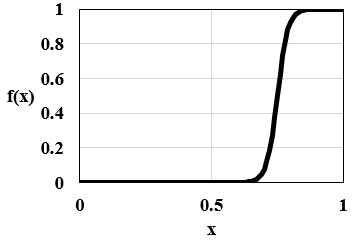}
	}
	\hfill
	\subfloat[a=0.75, b=30. \label{F22}]{%
	\includegraphics[scale=.53]{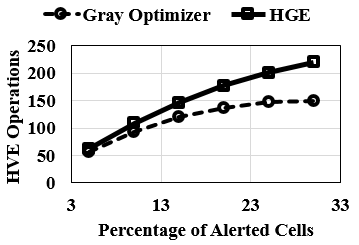}
	}
	\hfill
	\subfloat[a=0.75, b=30.\label{F23}]{%
	\includegraphics[scale=.53]{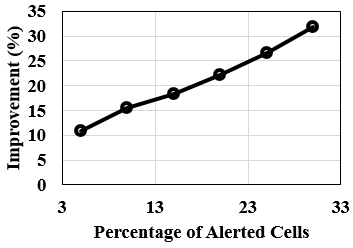}
	}
	\hfill
	\subfloat[a=0.5, b=10 (baseline performance).\label{F31}]{%
	\includegraphics[scale=.53]{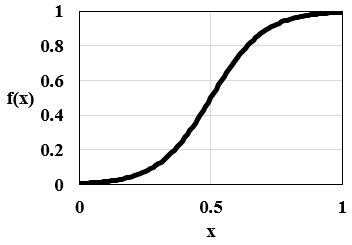}
	}
	\hfill
	\subfloat[a=0.5, b=10. \label{F32}]{%
	\includegraphics[scale=.53]{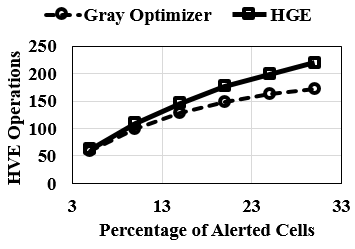}
	}
	\hfill
	\subfloat[a=0.5, b=10. \label{F33}]{%
	\includegraphics[scale=.53]{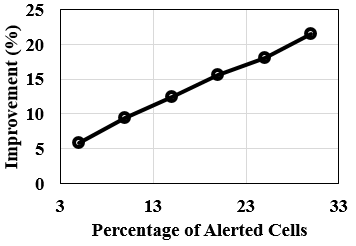}
	}
%		\hfill
%		\subfloat[a=0.5, b=30. \label{F41}.]{%
%		\includegraphics[scale=.5]{Figures/F41.png}
%		}
%		\hfill
%		\subfloat[a=0.5, b=30. \label{F42}.]{%
%		\includegraphics[scale=.5]{Figures/F42.png}
%		}
%		\hfill
%		\subfloat[a=0.5, b=30. \label{F43}.]{%
%		\includegraphics[scale=.5]{Figures/F43.png}
%		}
	\caption{Evaluation of GO, grid size = 100 cells.}
	\label{Fig_GrayOptimizer}
\end{figure*}

\subsection{Experimental Setup}\label{Initial Cell Probabilities}

We conduct our experiments on a $3.40$GHz core-i7 Intel processor with 8GB RAM running $64$-bit Windows $7$ OS. The code is implemented in Python, and we used the LogicMin Library~\cite{Minimization} for binary minimization of token expressions. We compare the proposed approaches (GO, MGSO and SGO) against the hierarchical Gray encoding technique from~\cite{ghinita2014efficient} (labeled {\em HGE}), the state-of-the-art in location alerts on HVE-encrypted data.

To model the probability of partition cells becoming alert zones, we use the sigmoid function $\mathcal{S}(x)= 1/(1+\exp^{-b(x-a)})$, where $a$ and $b$ are parameters controlling the function shape. The output value is between zero and one. The sigmoid function is a frequent model used in machine learning, and we chose it because in practice, the probability of individual cells becoming part of an alert zone can be computed using such a model built on a regions' map of features (e.g., type of terrain, building designation, point-of-sale information, etc). 
Parameter $a$ of the sigmoid controls the \textit{inflection} point of the curve, whereas $b$ controls the gradient.

\subsection{Gray Optimizer Evaluation}

GO is our core proposed algorithm to reduce the number of HVE operations required to support alert zones. Specifically, by \textit{HVE operations} we refer to the computation executed by the server to determine matches between tokens and encrypted user locations. Recall that, for each non-star item in a token, a number of expensive bilinear map operations are required. GO aims to minimize the number such non-star items in tokens by choosing an appropriate encoding of the domain.
Our comparison benchmark is the approach from~\cite{ghinita2014efficient} which uses a hierarchical quadtree structure to partition the data domain. We refer to this approach as \textit{HGE}, and we present our result as an improvement in terms of computation overhead compared with~\cite{ghinita2014efficient}.

\subsubsection{Improvement in HVE Operations}

Fig~\ref{Fig_GrayOptimizer} summarizes the evaluation results of GO for three logistic function parameter settings. The grid size is set to $100$ cells (recall from our earlier discussion that GO can only support relatively low granularities). Fig.~\ref{Fig_GrayOptimizer} shows the total number of bilinear pairings performed for a ciphertext-token pair. GO clearly outperforms the approach from~\cite{ghinita2014efficient}. The relative gain in performance of GO increases when the size of the alert zone increases (i.e., when there are more grid cells covered by the alert zone). This can be explained by the fact that a larger input set gives GO more flexibility to optimize the encoding and decrease the number of non-star entries in a token. In terms of percentage gains, GO can improve performance by up to 40\%, which is quite significant. Also, note that the gains are significant for all parameters of the sigmoid function used. In general, we identified that a higher $a$ value leads to more pronounced gains. This is an encouraging factor, because a higher $a$ corresponds to a more skewed probability case, where a relatively small number of cells are more likely to be included in an alert zone than others. In practice, one would expect that to be the case, since events that trigger alerts also tend to be concentrated over a relatively small area (e.g., very popular hotspots, certain facilities that present higher risks, like a chemical plant, etc.).

\begin{figure}[t]
\centering
\includegraphics[scale=.5]{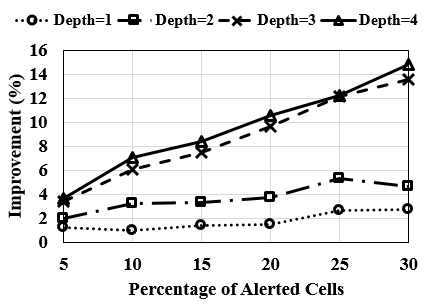}
\hspace{1em}
\centering
\caption{Performance evaluation of GO for varying depth (100 cells).}
\label{fig2}
\end{figure}

\begin{figure}[!t]
	\subfloat[Gray Optimizer. \label{fig4a}]{%
	\includegraphics[scale=.41]{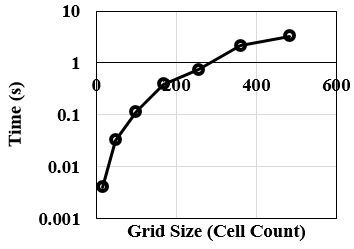}
	}
	\hfill
	\subfloat[MSGO.\label{fig4b}]{%
	\includegraphics[scale=.41]{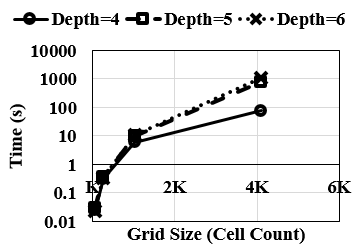}
	}
	\hfill
	\centering
	\subfloat[SGO.\label{fig4c}]{%
	\includegraphics[height=2.6cm, scale=.41]{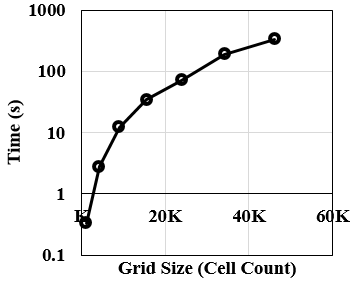}
	}
	\vspace{-10pt}
	\caption{Execution time.}
	\label{fig4}
\end{figure}

\subsubsection{Impact of Depth}

Recall that the reduction in computation achieved by GO depends on the depth at which the algorithm is run (GO works similar to a depth-first search graph algorithm).
In general, running the algorithm with a higher depth will produce better results in terms of performance gain at runtime (i.e., when matching is performed at the server), but it also requires a lot more computational time to compute a good encoding (which is a one-time cost). Fig.~\ref{fig2} captures the impact of depth on improvement. In this experiment, GO is executed on a single cell with different depths, and the remaining cells are assigned randomly (the experiment is specifically designed to show the effect of using lower depths on GO).
As expected, there is a clear increasing trend, with higher depths resulting in better improvement factors. However, after a sharp initial gain (illustrated by the large distance between the chart graphs corresponding to depths 2 and 3), the improvement stabilizes, and it may no longer be worth increasing the depth of the computation considerably (the gains are stabilizing between depths 3 and 4).

\subsubsection{Execution Time}

Fig.~\ref{fig4a} illustrates the execution time of GO. Recall that, the execution time of GO is influenced by the granularity of the grid (finer granularities increase execution time). The results show that GO can complete within a short execution time for smaller grid sizes; however, as the grid granularity increases,  there is a sharp increase in execution time. Therefore, GO may not be practical to apply for high granularity grids, and that is the main motivation behind our two variations, MSGO and SGO (which are evaluated next). Moreover, as the grid granularity increases, the length of cycles becomes larger, which will also result in numerical inaccuracies when executing GO. The execution time required by GO for values up to 600 cells is around 10 seconds. We observed that this value is the maximum number of cells for which GO performs reasonably; beyond this level, the algorithm is not suitable due to increased execution time and numerical inaccuracies associated with the calculation of probabilities for large cycles.

\begin{figure*}[!ht]

	\subfloat[a=0.75, b=10.\label{F5_11}]{%
	\includegraphics[scale=.53]{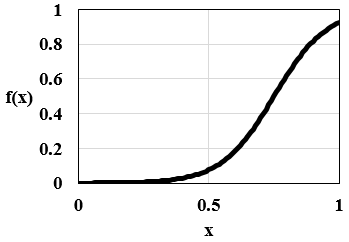}
	}
	\hfill
	\subfloat[a=0.75, b=10.\label{F5_12}]{%
	\includegraphics[scale=.53]{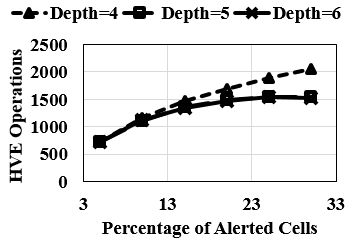}
	}
	\hfill
	\subfloat[a=0.75, b=10. \label{F5_13}]{%
	\includegraphics[scale=.53]{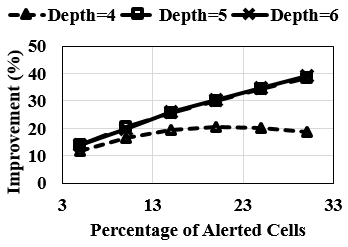}
	}
	\hfill
	\subfloat[a=0.75, b=30. \label{F5_21}]{%
	\includegraphics[scale=.53]{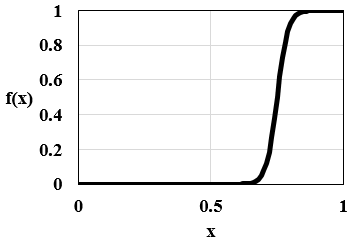}
	}
	\hfill
	\subfloat[a=0.75, b=30. \label{F5_22}]{%
	\includegraphics[scale=.53]{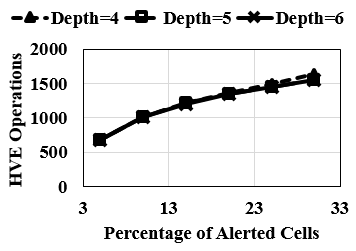}
	}
	\hfill
	\subfloat[a=0.75, b=30.\label{F5_23}]{%
	\includegraphics[scale=.53]{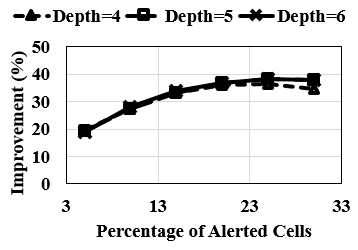}
	}
	\hfill
	\subfloat[a=0.5, b=10.\label{F5_31}]{%
	\includegraphics[scale=.53]{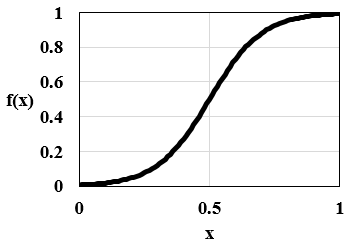}
	}
	\hfill
	\subfloat[a=0.5, b=10. \label{F5_32}]{%
	\includegraphics[scale=.53]{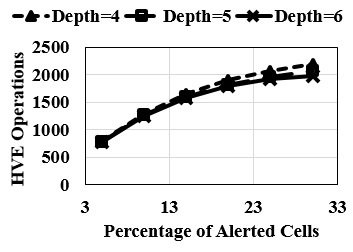}
	}
	\hfill
	\subfloat[a=0.5, b=10. \label{F5_33}]{%
	\includegraphics[scale=.53]{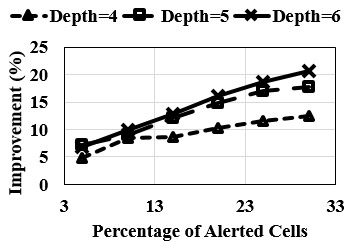}
	}
%		\hfill
%		\subfloat[a=0.5, b=30. \label{F5_41}.]{%
%		\includegraphics[scale=.58]{Figures/F5_41.png}
%		}
%		\hfill
%		\subfloat[a=0.5, b=30. \label{F5_42}.]{%
%		\includegraphics[scale=.58]{Figures/F5_42.png}
%		}
%		\hfill
%		\subfloat[a=0.5, b=30. \label{F5_43}.]{%
%		\includegraphics[scale=.58]{Figures/F5_43.png}
%		}
	\caption{Performance evaluation of MSGO (grid size = 1024 cells).}
	\label{fig_MSGO}
\end{figure*}

\subsection{Evaluation of GO Variations on Higher Granularity Grids}

As discussed previously, GO does not perform well when directly applied to high granularity grids. 
To improve the computational complexity of GO, we proposed two extensions of the algorithm, namely, MSGO and SGO. Next, we evaluate experimentally both these variations.% In this 

\subsubsection{MSGO}

\begin{figure*}[t]
	\subfloat[Grid size=10000 ]{%
	\includegraphics[scale=.57]{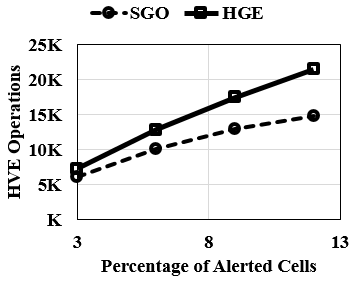}
	}
	\hfill
    \subfloat[Grid size=28900]{%
	\includegraphics[scale=.57]{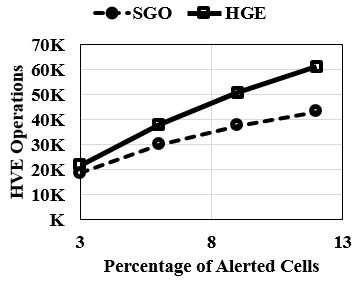}
	}
	\hfill
	\subfloat[Grid size=50625]{%
	\includegraphics[scale=.57]{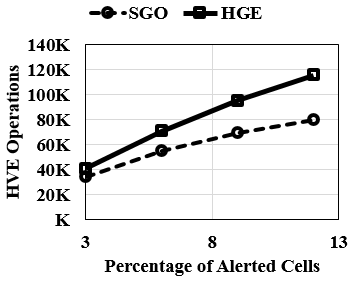}
	}
	\vspace{-10pt}
	\hfill
	\subfloat[Grid size=10000 ]{%
	\includegraphics[scale=.57]{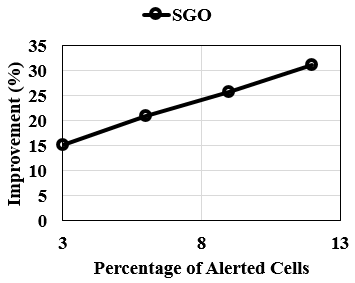}
	}
	\hfill
    \subfloat[Grid size=28900]{%
	\includegraphics[scale=.57]{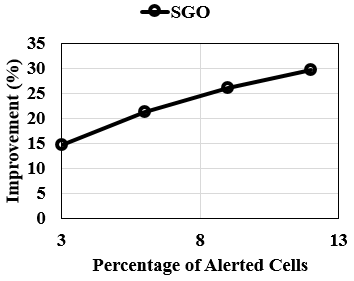}
	}
	\hfill
	\subfloat[Grid size=50625]{%
	\includegraphics[scale=.57]{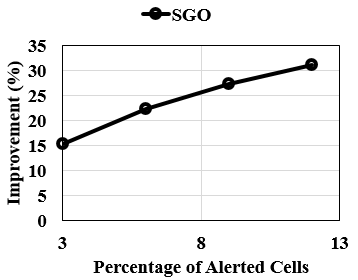}
	}
	\caption{SGO performance evaluation, varying grid size.}
	\label{fig6}
\end{figure*}	

\begin{figure*}[h!]
	\subfloat[GO]{%
	\includegraphics[scale=.57]{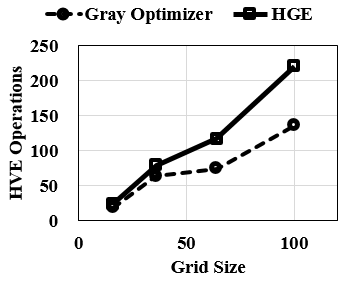}
	}
	\hfill
    \subfloat[MSGO, depth=4 ]{%
	\includegraphics[scale=.57]{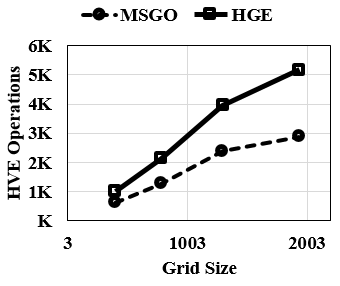}
	}
	\hfill
	\subfloat[SGO]{%
	\includegraphics[scale=.57]{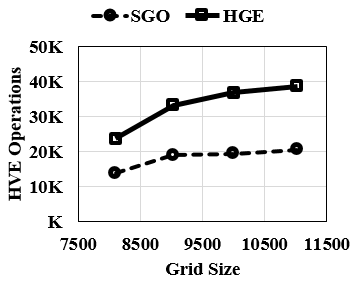}
	}
	\hfill
	\subfloat[GO]{%
	\includegraphics[scale=.57]{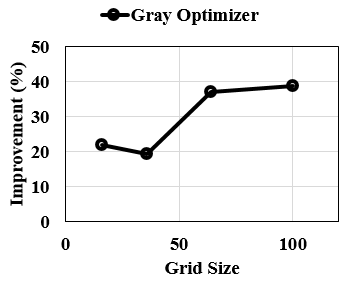}
	}
	\hfill
    \subfloat[MSGO, depth=4]{%
	\includegraphics[scale=.57]{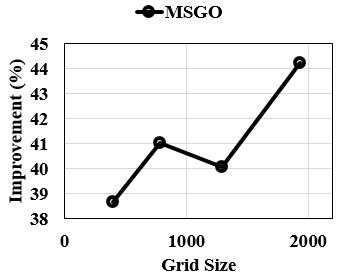}
	}
	\hfill
	\subfloat[SGO]{%
	\includegraphics[scale=.57]{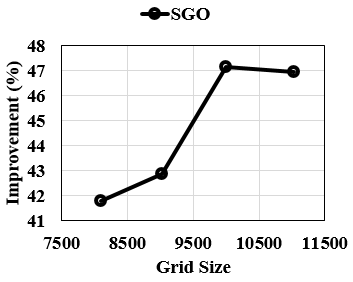}
	}
	\caption{Performance of algorithms, varying grid size, $30\%$ of cells on alert.}
	\label{Fig: Grid Size Variation}
\end{figure*}

Fig~\ref{fig_MSGO} illustrates the performance of MSGO compared to the HGE benchmark scheme from~\cite{ghinita2014efficient}. Unlike the single seed GO, we are able to evaluate the performance of MSGO for grids with much higher granularity (i.e., 1024 cells in this case). There is a similar trend in terms of gain as we have observed with GO, where larger alert zones provide more opportunities for advantageous encodings, and thus overall performance is improved (the percentage of HVE operations eliminated is higher). The relative gain obtained is very close to 50\% compared to the benchmark. Also, the absolute amount of improvement is better than for GO in all cases. This occurs due to the fact that MSGO can support higher-granularity grids, and in this setting there is more flexibility in choosing a good encoding (due to the larger number of cells, there are significantly more choices for our algorithm).  
As expected, increasing the depth of MSGO leads to higher improvement percentage, but the trade-off is a larger computation complexity.

Comparing Figs.~\ref{Fig_GrayOptimizer} and~\ref{fig_MSGO}, we remark that the MSGO algorithm obtains similar performance gains as the core algorithm GO for low granularity grids, but with a much lower computational overhead. For high granularity grids, GO cannot keep up in terms of computational overhead, whereas MSGO scales reasonably well, and it is able to still obtain significant improvements. One main reason is that MSGO no longer requires the calculation of probabilities of large cycles, avoiding numerical inaccuracies and reducing overall computational overhead. The complexity of the algorithm can be as low as $\mathcal{O}(n(\log_2n))$ depending on the chosen depth value, which provides a robust and efficient solution for reducing the number of HVE operations. 

The execution time of MSGO is illustrated in~Fig.~\ref{fig4b}. The graph indicates that even for a high level of granularity, such as $4,000$, the algorithm requires less than $15$ minutes to encode the grid, depending on the specified depth at the input. As expected, by increasing the depth of the algorithm, better performance can be achieved in terms of HVE operations, at the cost of higher computational overhead. The MSGO algorithm can be extended for an arbitrary number of cells on the grid, and also it may have various cluster sizes depending on the application.

\subsubsection{SGO}

Fig~\ref{fig6} illustrates the performance gain obtained by SGO. In this experiment, we focused on applying the algorithm to much larger number of cells, up to $50,625$ (which is equivalent to a $225 \times 225$ square grid). Similar to the MSGO algorithm, the improvement achieved by SGO occurs even when the alert zones are small. Since the overall number of cells is larger, the SGO algorithm has even more flexibility in choosing an advantageous encoding, resulting in strong performance gains. For example, at $9$\% ratio of alert cells, the SGO algorithm results in $25.8$, $26$, and $27.3$\% improvements for grid sizes of $10,000$, $28,900$, and $50,625$, respectively.

The execution time of SGO is shown in  Fig.~\ref{fig4c}. Even for very large grid sizes, such as $50,625$, the algorithm requires less than six minutes to encode the grid. Therefore, the system can be set to regularly update the probabilities and run the algorithm at six-minute intervals, if needed. To compare this time performance with GO, consider the maximum grid size for which the encoding can be computed within 60 seconds in each case. As shown in Fig.~\ref{fig4a}, this number corresponds to a grid size of $1200$ for GO, whereas in a similar time, SGO can be applied on the grid size of $22,000$ cells. Therefore, the SGO algorithm requires significantly lower computation overhead to execute compared with GO and even MSGO algorithms, while the performance gain in terms of HVE operations reductions is still solid.

Fig.~\ref{Fig: Grid Size Variation} presents the result of algorithms by fixing the percentage of alerted cells to $30\%$ and varying the grid size. It can be seen that the performance improvement of algorithms stays in a comparable margin for varying grid sizes. The slight fluctuation in graphs is due to two primary reasons (i) as all codewords have the same length, increasing the quantization level result in an addition of a bit to all codewords, and (ii) the initial probabilities are assigned to the cells in a random process based on the sigmoid activation function. 
%Hence, changing the grid size would require the generation of a new set of probabilities and a slight difference in performance improvement.

\begin{figure*}[h!]
	\subfloat[Grid size=100 \label{F6a}]{%
	\includegraphics[scale=.60]{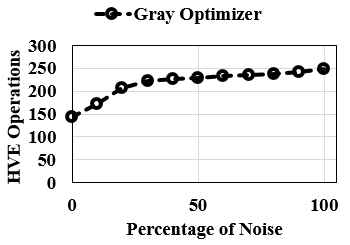}
	}
	\hfill
    \subfloat[Grid size=1024 \label{F6e}]{%
	\includegraphics[scale=.60]{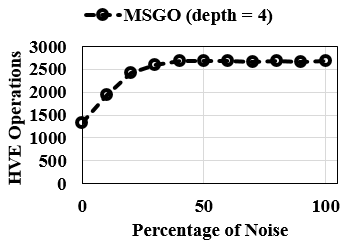}
	}
	\hfill
	\subfloat[Grid size=10000 \label{F6c}]{%
	\includegraphics[scale=.60]{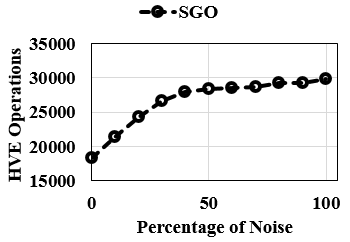}
	}
	%\vspace{-10pt}
	\hfill
	\subfloat[Grid size=100 \label{F6d}]{%
	\includegraphics[scale=.60]{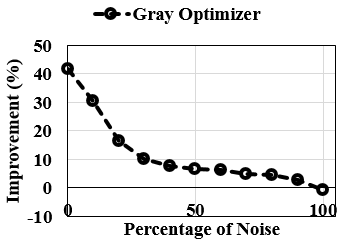}
	}
	\hfill
    \subfloat[Grid size=1024\label{F6b}]{%
	\includegraphics[scale=.60]{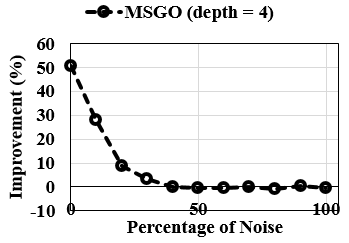}
	}
	\hfill
	\subfloat[Grid size=10000 \label{F6f}]{%
	\includegraphics[scale=.60]{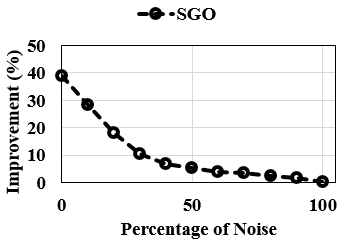}
	}
	\caption{Sensitivity analysis of algorithms, $40\%$ alerted cells.}
	\label{Fig: Sensitivity analysis}
\end{figure*}

\begin{figure*}[h]
	\subfloat[GO, grid size=100 ]{%
	\includegraphics[scale=.59]{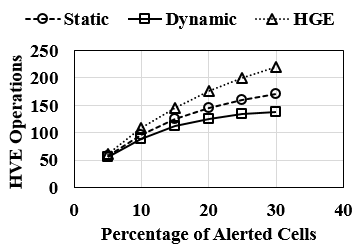}
	}
	\hfill
    \subfloat[MSGO, depth=4, grid size=1024 ]{%
	\includegraphics[scale=.59]{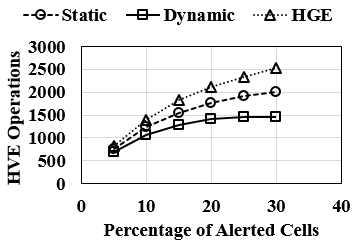}
	}
	\hfill
	\subfloat[SGO, grid size=10000 ]{%
	\includegraphics[scale=.59]{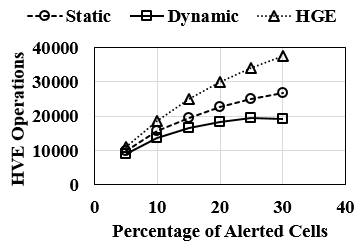}
	}
	\hfill
	\subfloat[GO, grid size=100 ]{%
	\includegraphics[scale=.60]{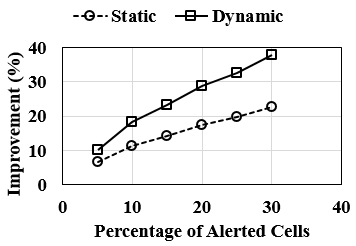}
	}
	\hfill
    \subfloat[MSGO, depth=4, grid size=1024 ]{%
	\includegraphics[scale=.60]{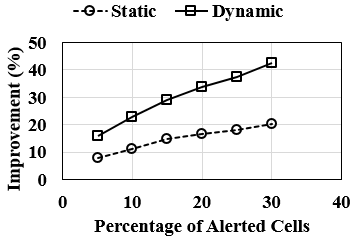}
	}
	\hfill
	\subfloat[SGO, grid size=10000 ]{%
	\includegraphics[scale=.60]{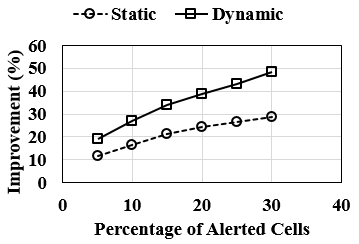}
	}
	\caption{Markov model versus static approach.}
	\label{Fig: dynamic approach}
\end{figure*}

\subsection{Imperfect Probabilities Information}

The knowledge of cell probabilities plays an important role in the reduction of HVE operations. These probabilities are input to GO and its extensions SGO and MSGO, used to find an enhanced encoding of space. Having imperfect initial cell probabilities can negatively impact the performance of algorithms by deviating the optimization result. Therefore, we aim at investigating the effect of imperfect initial probabilities on the improvement achieved compared with the previous work (HGE). This is done by the addition of noise to cell probabilities at the input of algorithms modeling the inaccuracies that might exist. Let us briefly illustrate how the addition of noise is conducted. Given the vector of cell probabilities:
$$\textrm{Probability of alert} = [p(v_1),\,p(v_2),\,...,\, p(v_n)],$$
each probability is added with an iid uniformly distributed random noise $n$ between $[0,u]$, where $u$ indicates the maximum noise value. For example, if the percentage of noise is $20\%$, this value is set to $0.2$, and the random noise is generated uniformly in the interval of $[0,0.2]$. Doing so, the transformed probabilities are acquired as 
$$[p(v_1)',\,p(v_2)',\,...,\, p(v_n)'],\, \textrm{where}\, \,p(v_i)' = \,p(v_i)+n_i. $$
Note that the values are considered to be cyclic between zero and one, i.e., if the noise of $0.5$ is added with a cell probability of $0.8$, the resulting value would be recorded as $0.3$. Hence, with $100\%$ of noise, it is expected that the numbers would be uniformly distributed.

Fig.~\ref{Fig: Sensitivity analysis} indicates the sensitivity of GO, MSGO, and SGO to imperfect probability values used as input. For each algorithm, the number of HVE operations required is shown as well as the improvement gained in the performance compared with the previous work. The $x$-axis represents the percentage of noise added to the perfect probability information varied between 0 to 100, and the $y$-axis indicates HVE operations required side by side to the improvement achieved. The percentage of alerted cells is set to $40\%$ in all graphs. 

The overall trend of reduction in the performance improvement by the addition of noise is consistent across all three algorithms. The improvement gained from algorithms stays at its highest when there exists no amount of noise at the input. The figure gradually drops as more noise is applied between zero to $50\%$, after which the performance improvement becomes almost negligible. 

As expected, in the case of maximum noise, no information is available regarding probabilities, and therefore, no further gain could be made with respect to HGE. Hence, at $100\%$ of noise, the number of HVE operations required from all algorithms converges to the HVE operations of HGE. The rate of sensitivity to imperfect information varies among algorithms. Looking at $10\%$ of noise, it can be seen that the drop in MSGO performance occurs at a higher rate than the other two algorithms with GO and SGO indicating $25\%$ loss in the performance against the loss of $40\%$ for the MSGO algorithm. Overall, MSGO shows a higher sensitivity compared with the GO and SGO algorithms.

\subsection{Dynamic Alert Zones}

So far, we evaluated techniques for static alert zones. Next, we measure the performance of our proposed technique for dynamic alert zones introduced in Section~\ref{Advance Modeling of Alert Zones}.

Fig.~\ref{Fig: dynamic approach} investigates the performance gain acquired by applying the proposed Markov model. The random path approach (Monte Carlo sampling) is used as the underlying method to compute the transition matrix's stationary distribution, minimizing the induced computation complexity on the system. The $x$-axis of the graphs shows the percentage of alert cells, and the $y$-axis represents the percentage of improvement as well as the number of HVE operations required. To distinguish between the two modeling approaches, performance improvement achieved by incorporating the Markov model is labeled as {\em dynamic}, and the scenario in which the time dependence is not considered is referred to as {\em static}. The experiment is designed by initializing both static and dynamic approaches with the same set of initial probabilities; however, the system would continue evolving in a uniformly distributed manner. Therefore, if there are $m$ outgoing edges from a state of the model, the corresponding probability is set to $1/m$. The aim is to see if the Markov model is able to capture the evolution of the system and how much improvement can be achieved with the gained information. As before, the value of the $a$ and $b$ are set to $0.75$ and $10$ with the termination probability of $0.4$.

Fig.~\ref{Fig: dynamic approach} shows that the dynamic method can predict well the evolution of alert zones, as the resulting encoding requires far fewer HVE operations. The performance gain achieved for all three of the algorithms is significant. The percentage of improvement is approximately $35\%$ to $50\%$, indicating more impact on GO compared to MSGO and SGO. 
%We emphasize that the performance gain depends on the underlying scenario and the extent to which the system evolves. Therefore, one should contemplate if the higher computation overhead of incorporating the Markov model is worth the improvement achieved in lowering the number of HVE operations.

\section{Related Work}
\subsection{Location Privacy.} Preserving the privacy of users in communication networks and online platforms has been one of the most challenging research problems in the past two decades. In the widely accepted scenario, users provide their location to service providers in exchange for location-based services they offer. The goal is to provide the service without user privacy being compromised by any of the parties involved. Early works to tackle this problem were focused on hiding or obfuscating user locations to achieve a privacy metric termed as $k$-anonymity. The location of a user is said to be $k$-anonymous if it is not distinguishable from at least $k-1$ other queried locations~\cite{sweeney2002k}.

In~\cite{kido2005anonymous}, the authors aim to provide $k$-anonymity by hiding the location of user among $k-1$ fake locations and requesting for desired services for all $k$ locations at the same time. The generation of such dummy locations based on a virtual grid or circle was considered in~\cite{niu2014privacy}. The authors in~\cite{niu2014achieving} conducted the selection of dummy locations predicated on the number of queries made on the map and aimed at increasing the entropy of $k$ locations in each set. In~\cite{cheng2006preserving}, random regions that enclose the user locations were introduced to bring uncertainty in the authentication of user locations. Unfortunately, fake locations can be revealed particularly in trajectories and with the existence of prior knowledge about the map and users.

Later on, approaches based on {\em Cloaking Regions (CRs)} proposed by~\cite{gruteser2003anonymous} gained momentum in the literature. The principal idea behind this method is to use a trusted anonymizer that clusters $k$ real user locations and query the area they are enclosed by to retrieve points of interest. Doing so, CRs aim to achieve $k$-anonymity for users and preserve their privacy. This approach is partially effective when snapshots of trajectories are considered, but once users are seen in trajectories, their location privacy would be severely at risk~\cite{shaham2020privacy}. Even for individual snapshots, it must be noted that a coarse area of real locations is released to the service provider, which could threaten the location privacy of users. Moreover, the CR-based approaches are susceptible to inference attacks predicated on the background knowledge or so-called side information. One such side information is the knowledge about the number of queries made on different locations of the map~\cite{niu2014achieving}.  

More recently, a model for privacy preservation in statistical databases termed as {\em differential privacy} was developed in~\cite{dwork2006calibrating}. The metric provides a promising prospect for aggregate queries; however, it is not suitable for private retrieval of specific data from datasets. Closer to HVE approach, a private information protocol was proposed in~\cite{ghinita2008private}. The PIR technique is based on cryptography and shown to be secure for private retrieval of information. Despite the promising results, there exists an assumption behind PIR approach that the user already knows about the points of interest. Therefore, PIR is not suitable for location-based alert systems as users are not aware of alert zone whereabouts.

\subsection{Searchable Encryption.} Originated from works such as~\cite{song2000practical}, the concept of search encryption was proposed to provide a secure cryptographic search of keywords.  Initially, only the exact matches of keywords were supported and later on the approach was extended for comparison queries in~\cite{boneh2006fully}, and to subset queries and conjunctions of equality in~\cite{boneh2007conjunctive}. The authors in~\cite{boneh2007conjunctive} also proposed the concept of HVE, used as the underlying tool to provide a secure location-based alert system. This approach and its extension in~\cite{blundo2009private} preserves the privacy of encrypted messages and tokens with the overhead of high computational complexity. The authors in~\cite{ghinita2014efficient} introduced and adopted the HVE for location-based alert systems, conducting the predicate match at a trusted provider, preserving the privacy of encrypted messages as well as tokens. Despite the promising results of the approach for privacy preservation in location-based alert systems, further reduction of computational overhead is necessary to increase the practicality.

\begin{table*}[th]
\caption{Summary of the proposed algorithms.}
\centering
\begin{tabular}{|>{\centering\arraybackslash}m{1.6cm} || >{\centering\arraybackslash}m{1.6cm} | >{\centering\arraybackslash}m{2.5 cm} |>{\centering\arraybackslash}m{2.7cm} |>{\centering\arraybackslash}m{1.2cm} |>{\centering\arraybackslash}m{5cm} |}
 \hline Algorithms & Possible granularity& Complexity & Recommended depth & Number of seeds& Application\\
 \hline Gray Optimizer & Low &  $\mathcal{O}(2n(\log_2n)) + n^{log_23} - 2n+1$ & $|\log_2n|$ & one &The core algorithm used in MSGO and SGO, which is suitable for low granularity grids due to high complexity.\\
 \hline MSGO & Medium to high&  The minimum of $\mathcal{O}(n(\log_2n))$ depending on the depth variable& Depth=4 when the clusters are of equal length, otherwise it depends on the application & Equal to number of clusters &Advantageous in scenarios where the optimization is required for particular subsets of the grid. For example, the service is needed for two major non-overlapping organizations.\\
 \hline SGO  & High &  $\mathcal{O}(n(\log_2n))$ &One & $n$ &SGO requires the lowest complexity, and it is advantageous in scenarios where the aim is to optimize the number of HVE operations for all the cells with equal significance.\\
 \hline
\end{tabular}
\label{t1}
\end{table*}

\section{Conclusion}

We proposed a family of techniques to reduce the computational overhead of HVE predicate evaluation in location-based alert systems. Specifically, we used graph embeddings to find advantageous domain space encodings that help reduce the required number of expensive HVE operations. Our heuristic solutions provide a significant improvement in computation compared to existing work, and they can scale to domain partitionings of fine granularity. In addition, we studied how to extend these techniques to work for the challenging setting of dynamic alert zones. Table~\ref{t1} summarizes the properties of the proposed approaches.

In future work, we will focus on deriving cost models and strategies to reduce the HVE overhead based on workload-specific requirements. Certain families of tasks may exhibit specific patterns of operations, which can be taken into account to optimize HVE matching performance, as well as to re-use computation. We will also investigate extending the graph embedding approach to other types of searchable encryption, beyond HVE (e.g., Inner Product Evaluation), which exhibit different types of internal algebraic operations. 

{\bf Acknowledgment.} This research has been funded in part by NSF grants IIS-1910950 and  IIS-1909806, the USC Integrated Media Systems Center (IMSC), and unrestricted cash gifts from Microsoft and Google . Any opinions, findings, and conclusions or recommendations expressed in this material are those of the author(s) and do not necessarily reflect the views of any of the sponsors such as the National Science Foundation.
%Note: Cyrus Shahabi receives consulting income from Google for a different unrelated project. 

\bibliographystyle{abbrv} 
\bibliography{Ref_HVE}   % name your BibTeX data base

\appendix
\section{Primer on HVE Encryption}
\label{sec:app}

{\em Hidden Vector Encryption (HVE)} \cite{boneh2007conjunctive} is a searchable encryption system that supports predicates in the form of conjunctive equality, range and subset queries. Search on ciphertexts can be performed with respect to a number of {\em index attributes}. HVE represents an attribute as a bit vector (each element has value $0$ or $1$), and the search predicate as a {\em pattern} vector where each element can be $0$, $1$ or '*' that signifies a wildcard (or ``don't care'') value. Let $l$ denote the HVE {\em width}, which is the bit length of the attribute, and consequently that of the search predicate. A predicate evaluates to $True$ for a ciphertext $C$ if the attribute vector $I$ used to encrypt $C$ has the same values as the pattern vector of the predicate in all positions that are not '*' in the latter. Fig.~\ref{fig:math-nonmatch} illustrates the two cases of {\em Match} and {\em Non-Match} for HVE.

HVE is built on top of a symmetrical bilinear map of composite order \cite{boneh2007conjunctive}, which is a function $e : \mathbb{G} \times \mathbb{G} \rightarrow \mathbb{G}_T$ such that $\forall a,b \in G$ and $ \forall u,v \in \mathbb{Z}$ it holds that $e(a^u,b^v)=e(a,b)^{uv}$. $\mathbb{G}$ and $\mathbb{G}_T$ are cyclic multiplicative groups of composite order $N=P\cdot Q$ where $P$ and $Q$ are large primes of equal bit length. We denote by $\mathbb{G}_p$, $\mathbb{G}_q$ the subgroups of $\mathbb{G}$ of orders $P$ and $Q$, respectively. Let $l$ denote the HVE {\em width}, which is the bit length of the attribute, and consequently that of the search predicate. HVE consists of the following phases:

{\bf Setup.} The $TA$ generates the public/secret ($PK$/$SK$) key pair and shares $PK$ with the users. $SK$ has the form:
$$SK = ( g_q \in \mathbb{G}_q,\quad a \in \mathbb{Z}_p,\quad \forall i \in [1..l]: u_i,h_i, w_i, g, v \in \mathbb{G}_p )$$
To generate $PK$, the $TA$ first chooses at random elements \(R_{u,i}, R_{h,i}\), \(R_{w,i} 
\in \mathbb{G}_q, \forall i \in [1..l]\)  and \(R_v \in \mathbb{G}_q\). Next, $PK$ is determined as:
$$PK = (g_q,\quad V=vR_v,\quad A=e(g,v)^a,\quad$$
$$\forall i \in [1..l]: U_i=u_iR_{u,i},\quad  H_i=h_iR_{h,i},\quad  W_i=w_iR_{w,i})$$

{\bf Encryption} uses $PK$ and takes as parameters index attribute $I$ and message $M \in \mathbb{G}_T$. The following random elements are generated: \(Z, Z_{i,1}, Z_{i,2} \in \mathbb{G}_q\) and \(s \in \mathbb{Z}_n\). Then, the ciphertext is: 
$$C = (C^{'}= MA^s,\quad C_0=V^sZ, \quad $$
$$\forall i \in [1..l]: C_{i,1} = (U^{I_i}_iH_i)^sZ_{i,1}, \quad C_{i,2} = W^{s}_iZ_{i,2} )$$

{\bf Token Generation.} Using $SK$, and given a search predicate encoded as pattern vector \(I_{*}\), the TA generates 
a search token $TK$ as follows: let \(J\) be the set of all indices $i$ where \(I_{*}[i] \neq *\).
TA randomly generates \(r_{i,1}\) and \(r_{i,2} \in \mathbb{Z}_p, \forall i \in J\). 
Then
$$TK=(I_*, K_0 = g^a\prod_{i \in J}(u^{I_{*}[i]}_ih_i)^{r_{i,1}}w^{r_{i,2}}_i, \quad$$
$$ \forall i \in [1..l]: K_{i,1} = v^{r_i,1},\quad K_{i,2} = v^{r_i,2})$$

{\bf Query} is executed at the server, and evaluates if the predicate represented by $TK$ holds for ciphertext $C$. The server attempts to determine the value of \(M\) as 
\begin{equation}
M = C^{'}{/} (e(C_0,K_0) {/} \prod_{i \in J} e(C_{i,1},K_{i,1}) e(C_{i,2},K_{i,2}) \label{eq:query}
\end{equation}
If the index $I$ based on which $C$ was computed satisfies $TK$, then the actual value of \(M\) is returned, otherwise a special number which is not in the valid message domain (denoted by $\bot$) is obtained.

\end{document}